\documentclass[a4paper,english,numberwithinsect,final]{lipics-v2016}
\pdfoutput=1

\usepackage[utf8]{inputenc}
\usepackage[T1]{fontenc}

\EventEditors{Roland Meyer and Uwe Nestmann}
\EventNoEds{2}
\EventLongTitle{28th International Conference on Concurrency Theory (CONCUR 2017)}
\EventShortTitle{CONCUR 2017}  
\EventAcronym{CONCUR}
\EventYear{2017}
\EventDate{September 5--8, 2017}
\EventLocation{Berlin, Germany}
\EventLogo{}
\SeriesVolume{85}
\ArticleNo{28}

\usepackage[notcite,notref]{showkeys}

\usepackage{color}

\usepackage[utf8]{inputenc}
\usepackage{amsmath}
\usepackage{tikz-cd}
\usepackage{etoolbox}
\usepackage{xspace}
\usepackage{mathtools}
\usepackage{amsthm}
\usepackage{thmtools}
\usepackage[noend]{algpseudocode}
\usepackage[inline]{enumitem}
\newcommand{\unnicefrac}[2]{\ensuremath{#1\mkern-1.5mu/\mkern-1.5mu{#2}}}
\usepackage{rotating}
\usepackage{array}   
\usepackage{subcaption}
\usepackage{fdsymbol}
\usepackage{dsfont}
\usepackage{environ}
\newcolumntype{L}{>{\(}l<{\)}} 

\definecolor{myorange}{HTML}{DA7200}
\definecolor{myblue}{HTML}{0050DC}
\definecolor{lipicsYellow}{rgb}{0.99,0.78,0.07}

\hypersetup{
    colorlinks=true,
    linkcolor=black,
    citecolor=black,
    filecolor=black,
    urlcolor=black,
    pdftitle={Efficient Coalgebraic Partition Refinement},
    pdfauthor={Ulrich Dorsch, Stefan Milius, Lutz Schröder, Thorsten Wißmann},
    pdfkeywords={markov chains, deterministic finite automata, partition
      refinement, generic algorithm, paige-tarjan algorithm, transition systems},
    pdfduplex={DuplexFlipLongEdge},
}

\usepackage{seqsplit}
\usepackage{xstring}

\newcommand\noshowkeys{\def\hideNextShowKeysLabel{test}}
\makeatletter
\renewcommand*\showkeyslabelformat[1]{%
\@ifundefined{hideNextShowKeysLabel}{%
\noexpandarg%
\StrSubstitute{#1}{ }{\textvisiblespace}[\TEMP]%
\parbox[t]{\marginparwidth}{\raggedright\normalfont\small\ttfamily\(\{\){\color{red!50!black}\expandafter\seqsplit\expandafter{\TEMP}}\(\}\)}%
}{}
}
\makeatother

\patchcmd{\thmhead}{(#3)}{#3}{}{}

\usetikzlibrary{calc}
\usetikzlibrary{fit}
\usetikzlibrary{backgrounds}
\usetikzlibrary{scopes}
\usetikzlibrary{external}
\tikzset{external/up to date check=diff}

\tikzsetexternalprefix{diagrams/}
\makeatletter
\tikzifexternalizing{%
\def\@enddocumenthook{}%
}{}
\makeatother

\def\temp{&} \catcode`&=\active \let&=\temp

\tikzcdset{
arrow style=tikz,
diagrams={>={Straight Barb[scale=0.8]}}
}

\newcommand{\mytikzcdcontext}[2]{
  \begin{tikzpicture}[baseline=(mainnode.base)]
    \node (mainnode) [inner sep=0, outer sep=0] {\begin{tikzcd}[#2]
        #1
    \end{tikzcd}};
  \end{tikzpicture}}

\NewEnviron{mytikzcd}[1][]{%
\def\myargs{#1}%
\edef\mydiagram{\noexpand\mytikzcdcontext{\expandonce\BODY}{\expandonce\myargs}}%
\mydiagram%
}

\makeatletter
\@ifundefined{BC}{%

}{}
\makeatother

\newcommand{\Sorts}{\mathcal{S}}

\newcommand{\gap}{\mathbin{\_}}

\newcommand{\CO}{\mathcal{O}}
\newcommand{\op}[1]{\textsf{\upshape #1}}
\newcommand{\inj}{\op{in}}
\newcommand{\inl}{\op{inl}}
\newcommand{\inr}{\op{inr}}
\newcommand{\old}{{\op{old}}}
\newcommand{\id}{\op{id}}

\newcommand{\geZero}{>^{\!\!\smash{?}}}
\newcommand{\leZero}{{\ }^{?}\!\!\mathord{<}}
\renewcommand{\Im}{\mathop{\op{Im}}}
\renewcommand{\ker}{\mathop{\op{ker}}}
\newcommand{\out}{\pi}

\newcommand{\Pot}{\mathcal{P}}
\newcommand{\Bag}{\mathcal{B}}

\newcommand{\partialto}{\rightharpoonup}
\newcommand{\fpair}[1]{\langle #1\rangle}
\newcommand{\Set}{\ensuremath{\mathsf{Set}}\xspace}

\newcommand{\Potf}{\ensuremath{\mathcal{P}_{\mathrm{f}}}\xspace}
\newcommand{\Dist}{\ensuremath{\mathcal{D}}\xspace}
\newcommand{\Bagf}{\ensuremath{\mathcal{B}_{\mathrm{f}}}\xspace}
\newcommand{\C}{\ensuremath{\mathcal{C}}\xspace}
\newcommand{\D}{\ensuremath{\mathcal{D}}\xspace}
\newcommand{\E}{\ensuremath{\mathcal{E}}\xspace}
\newcommand{\Z}{\ensuremath{\mathds{Z}}\xspace}
\newcommand{\N}{\ensuremath{\mathds{N}}\xspace}
\newcommand{\B}{\ensuremath{\mathds{B}}\xspace}
\newcommand{\R}{\ensuremath{\mathds{R}}\xspace}

\newcommand{\Coalg}{\ensuremath{\mathsf{Coalg}}\xspace}
\newcommand{\Inputs}{\ensuremath{A}}
\newcommand{\Leftblock}{\mathcal{L}}
\newcommand{\Middleblock}{\mathcal{M}}
\renewcommand{\gets}{\ensuremath{\xspace\mathop{:=}}}
\newcommand{\groupsum}{\raisebox{2pt}{$\scriptstyle\sum$}}
\renewcommand{\subsetneq}{\subsetneqq}

\tikzstyle{inlinecd}=[column sep = 5mm, row sep = 5mm]
\tikzstyle{weakly}=[dash pattern=on 1pt off 1pt]
\tikzstyle{shiftarr}=[
        rounded corners,%
        to path={--([#1]\tikztostart.center)
                     -- ([#1]\tikztotarget.center) \tikztonodes
                     -- (\tikztotarget)},
]

\newsavebox{\mypullbackcorner}%
\sbox{\mypullbackcorner}{%
\tikzexternaldisable%
\begin{tikzpicture}
    \draw[-] (0,0) -- (.5em,.5em) -- (0,1em);
\end{tikzpicture}%
\tikzexternalenable}
\newcommand{\pullbackangle}[2][]{\arrow[phantom,to path={
                     -- ($ (\tikztostart)!1cm!#2:([xshift=8cm]\tikztostart) $)
                        node[anchor=west,pos=0.0,rotate=#2,
                        inner xsep = 0]
                        {\begin{tikzpicture}[minimum
                        height=1mm,baseline=0,#1]
    \draw[-] (0,0) -- (.5em,.5em) -- (0,1em);
                        \end{tikzpicture}}}]{}}
\newcommand{\descto}[2]{\arrow[phantom]{#1}{\text{\footnotesize{}#2}}}

\DeclareUnicodeCharacter{D7}{\ensuremath{\times}}

\makeatletter
  \edef\thetheoremsuffix{\@thmcountersep\@thmcounter{theorem}}
  
\makeatother
\theoremstyle{plain}
\newtheorem*{theorem*}{Theorem}
\newtheorem{corollary}[theorem]{Corollary}

\newtheorem{lemma}[theorem]{Lemma}
\newtheorem{proposition}[theorem]{Proposition}

\theoremstyle{definition}
\newtheorem{algorithm}[theorem]{Algorithm}
\newtheorem{assumption}[theorem]{Assumption}

\newtheorem{definition}[theorem]{Definition}
\newtheorem{examples}[theorem]{Examples}
\newtheorem{example}[theorem]{Example}
\newtheorem{notation}[theorem]{Notation}
\newtheorem{remark}[theorem]{Remark}

\theoremstyle{remark}

\numberwithin{theorem}{section}
\numberwithin{equation}{section}

\def\epito{\twoheadrightarrow}
\def\monoto{\rightarrowtail}

\def\upa{\mathord{\uparrow}}

\newcommand{\takeout}[1]{\empty}

\usepackage{microtype}

\addto\extrasenglish{
}

\title{Efficient Coalgebraic Partition Refinement\footnote{Full version with all proof
    details available at \url{http://arxiv.org/abs/1705.08362}
    \newline
    This work forms part of the DFG-funded project COAX (MI~717/5-1 and SCHR~1118/12-1)
    }
  }
\titlerunning{Efficient Coalgebraic Partition Refinement}

\author{Ulrich Dorsch
}
\author{Stefan Milius
}
\author{Lutz Schr\"oder%
}
\author{Thorsten \rlap{Wi\ss{}mann}%
}
\affil{
  Friedrich-Alexander-Universit\"{a}t Erlangen-N\"{u}rnberg, Germany}
\authorrunning{U.~Dorsch, S.~Milius, L.~Schr\"oder and T.~Wi\ss{}mann}
\Copyright{Ulrich Dorsch, Stefan Milius, Lutz Schr\"oder and Thorsten Wi\ss{}mann}


\begin{document}
\maketitle
\begin{abstract}
  We present a generic partition refinement algorithm that quotients
  coalgebraic systems by behavioural equivalence, an important task in
  reactive verification; coalgebraic generality implies in particular
  that we cover not only classical relational systems but also various
  forms of weighted systems. Under assumptions on the type functor
  that allow representing its finite coalgebras in terms of nodes and
  edges, our algorithm runs in time $\mathcal{O}(m\cdot \log n)$ where $n$ and
  $m$ are the numbers of nodes and edges, respectively. Instances of
  our generic algorithm thus match the runtime of the best known
  algorithms for unlabelled transition systems, Markov chains, and
  deterministic automata (with fixed alphabets), and improve the best
  known algorithms for Segala systems.
\end{abstract}

%
%
\section{Introduction}

\emph{Minimization under bisimilarity} is the task of identifying all
states in a reactive system that exhibit the same
behaviour. Minimization appears as a subtask in state space reduction
(e.g.~\cite{BlomO05}) or non-interference
checking~\cite{MeydenZ07}. The notion of bisimulation was first
defined for relational systems~\cite{Benthem77,Milner80,Park81}; it
was later extended to other system types including probabilistic
systems~\cite{LarsenS91,DesharnaisEA02} and weighted
automata~\cite{Buchholz08}. In fact, the importance of minimization
under bisimilarity appears to increase with the complexity of the
underlying system type. E.g., while in LTL model checking,
minimization drastically reduces the state space but, depending on the
application, does not necessarily lead to a speedup in the overall
balance~\cite{FislerV02}, in probabilistic model checking,
minimization under strong bisimilarity does lead to substantial
efficiency gains~\cite{KatoenEA07}.

The algorithmics of minimization, often referred to as \emph{partition
  refinement} or \emph{lumping}, has received a fair amount of
attention. Since bisimilarity is a greatest fixpoint, it is more or
less immediate that it can be calculated in polynomial time by
approximating this fixpoint from above following Kleene's fixpoint
theorem. In the relational setting, Kanellakis and
Smolka~\cite{KanellakisS90} introduced an algorithm that in fact runs
in time $\mathcal{O}(nm)$ where $n$ is the number of nodes and $m$ is
the number of transitions. An even more efficient algorithm running in
time~$\mathcal{O}(m\log n)$ was later described by Paige and
Tarjan~\cite{PaigeTarjan87}; this bound holds even if the number of
action labels is not fixed~\cite{Valmari09}. Current algorithms
typically apply further optimizations to the Paige-Tarjan algorithm,
thus achieving better average-case behaviour but the same worst-case
behaviour~\cite{DovierEA04}. Probabilistic minimization has undergone
a similarly dynamic development~\cite{BaierEM00,CattaniS02,ZhangEA08},
and the best algorithms for minimization of Markov chains now have the
same $\CO(m\log n)$ run time as the relational Paige-Tarjan
algorithm~\cite{HuynhTian92,DerisaviEA03,ValmariF10}.  Using ideas
from abstract interpretation, Ranzato and Tapparo~\cite{RanzatoT08}
have developed a relational partition refinement algorithm that is
generic over \emph{notions of process equivalence}. As instances, they
recover the classical Paige-Tarjan algorithm for strong bisimilarity
and an algorithm for stuttering equivalence, and obtain new algorithms
for simulation equivalence and for a new process equivalence.


In this paper we follow an orthogonal approach and provide a generic
partition refinement algorithm that can be instantiated for many
different \emph{types} of systems (e.g.\ nondeterministic,
probabilistic, weighted). We achieve this by methods of
\emph{universal coalgebra}~\cite{Rutten00}. That is, we encapsulate
transition types of systems as endofunctors on sets (or a more general
category), and model systems as coalgebras for a given type functor.

Our work proceeds on several levels of abstraction. On the most abstract level
(Section~\ref{sec:cat}) we work with coalgebras for a monomorphism-preserving
endofunctor on a category with image factorizations. Here we present a quite
general category-theoretic partition refinement algorithm, and we prove its
correctness. The algorithm is parametrized over a \op{select} routine that
determines which observations are used to split blocks of states; the corner case where
all available observations are used yields known coalgebraic final chain
algorithms, e.g.~\cite{KonigKupper14}.

Next, we present an optimized version of our algorithm
(Section~\ref{sec:opti}) that needs more restrictive conditions to
ensure correctness; specifically, we need to assume that the type
endofunctor satisfies a condition we call \emph{zippability} in order
to allow for incremental computation of partitions. This property
holds, e.g., for all polynomial endofunctors on sets and for the type
functors of labelled and weighted transition systems, but not for all
endofunctors of interest. In particular, zippable functors fail to be
closed under composition, as exemplified by the double covariant
powerset functor $\Pot\Pot$ on sets, for which the optimized algorithm
is in fact incorrect. However, it turns out that obstacles of this
type can be removed by moving to multi-sorted
coalgebras~\cite{SchroderPattinson11}, so we do eventually obtain an
efficient partition refinement algorithm for coalgebras of composite
functors, including $\Pot\Pot$-coalgebras as well as (probabilistic)
Segala systems~\cite{Segala95}.

Finally, we analyse the run time of our algorithm
(Section~\ref{sec:efficient}). To this end, we make our algorithm
parametric in an abstract \emph{refinement interface} to the type
functor, which encapsulates the incremental calculation of partitions
in the optimized version of the algorithm. We show that if the
interface operations can be implemented in linear time, then the
algorithm runs in time $\CO(m \log n)$, where $n$ is the number of
states and $m$ the number of `edges' in a syntactic encoding of the
input coalgebra. We thus recover the most efficient known algorithms
for transition systems (Paige and Tarjan~\cite{PaigeTarjan87}) and for
weighted systems (Valmari and Franceschinis~\cite{ValmariF10}). Using
the mentioned modularity results, we also obtain an
$\CO((m+n) \log (m+n))$ algorithm for Segala systems, to our knowledge
a new result (more precisely, we improve an earlier bound established
by Baier, Engelen, and Majster-Cederbaum~\cite{BaierEM00}, roughly
speaking by letting only non-zero probabilistic edges enter into the
time bound). The algorithm and its analysis apply also to generalized
polynomial functors on sets; in particular, for the functor
$2 \times (-)^A$, which models deterministic finite automata, we
obtain the same complexity $\CO(n\log n)$ as for Hopcroft's classical
minimization algorithm for a fixed alphabet $A$~\cite{Hopcroft71,
  Knuutila2001, Gries1973}.

\vspace{-1mm}
\section{Preliminaries}
\vspace{-2mm}
\label{sec:prelim}

We assume that readers are familiar with basic category
theory~\cite{joyofcats}. For the convenience of the reader we recall
some concepts that are central for the categorical version of the
algorithm.

\begin{notation} \label{pullbackNotation} The terminal object is
  denoted by $1$, with unique arrows $!:A\to 1$, and the product of
  objects $A$, $B$ by
  $A \xleftarrow{\out_1} A×B \xrightarrow{\out_2} B$. Given
  $f: D\to A$ and $g: D\to B$, the morphism induced by the universal
  property of the product $A \times B$ is denoted by
  $\fpair{f,g}:D\to A\times B$.  The \emph{kernel} $\ker f$ of a
  morphism $f$ is the pullback of $f$ along itself.  We write $\epito$
  for regular epimorphisms (i.e.~coequalizers), and $\rightarrowtail$ for monomorphisms.
\end{notation}

\noindent Kernels allow us to talk about equivalence relations in
a category. In particular in \Set, there is a bijection between
kernels and equivalence relations in the usual sense: For a map
$f: D \to A$, $\ker f = \{(x,y) \mid fx = fy\}$ is the equivalence
relation induced by $f$. Generally, relations (i.e.\ jointly monic
spans of morphisms) in a category are ordered by inclusion in the
obvious way. We say that a kernel $K$ is \emph{finer} than a
kernel~$K'$ if~$K$ is included in~$K'$. We use intersection $\cap$
and union $\cup$ of kernels for meets and joins in the inclusion
ordering on \emph{relations} (not equivalence relations or kernels);
in this notation, $\ker \fpair{f,g} =\ker f\cap\ker g$.
In \Set, a map $f:D\to A$ factors through the partition
$\unnicefrac{D}{\ker f}$ induced by its kernel, via the map
$[-]_f: D\twoheadrightarrow \unnicefrac{D}{\ker f}$ taking equivalence classes
\begin{equation*}
  [x]_f := \{ x'\in D\mid fx = fx'\} = \{ x' \in D\mid (x,x') \in \ker f\}.
\end{equation*}
Well-definedness of functions on $\unnicefrac{D}{\ker f}$ is determined
precisely by the universal property of $[-]_f$ as a coequalizer of $\ker
f\rightrightarrows D$. In particular, $f$ induces an injection
$\unnicefrac{D}{\ker f}\rightarrowtail A$; together with $[-]_f$, this is the
factorization of $f$ into a regular epimorphism and a monomorphism.
Categorically, this is captured by the following assumptions.
\begin{assumption}
  \label{ass:C}
  We assume throughout that $\C$ is a finitely complete category that
  has coequalizers and \emph{image factorizations}, i.e.~every
  morphism $f$ has a factorization $f = m \cdot e$ as a regular
  epimorphism $e$ followed by a monomorphism $m$.  We call the
  codomain of $e$ the \emph{image} of $f$, and denote it by
  $\unnicefrac{D}{\ker f}$.  Regular epis in $\C$ are
  closed under composition and right
  cancellation~\cite[Prop.~14.14]{joyofcats}.
\end{assumption}
%

\begin{examples}
  Examples of categories satisfying~\autoref{ass:C} abound. In
  particular, every regular category with coequalizers satisfies our
  assumptions. The category $\Set$ of sets and maps is, of course, regular.
  Every topos is regular, and so is every finitary variety, i.e.~a category of
  algebras for a finitary signature satisfying given equational axioms
  (e.g.~monoids, groups, vector spaces etc.). Posets and topological spaces fail
  to be regular but still satisfy our assumptions. If $\C$ is regular, so is the
  functor category $\C^\E$ for any category $\E$.
\end{examples}

\noindent For a set $\Sorts$ of sorts, the category $\Set^\Sorts$ of
$\Sorts$-sorted sets has $\Sorts$-tuples of sets as objects. We write
$\chi_S: X\to 2$ for the characteristic function of a subset
$S\subseteq X$, i.e.\ for $x\in X$ we have $\chi_S(x)=1$ if $x\in S$
and $\chi_S(x)=0$ otherwise. We will also use a three-valued version:
\begin{definition}\label{D:chi}
    For $S\subseteq C\subseteq X$, define
    \(
    \chi_S^C: X\to 3
    \) by
    \(   C\not\owns x \mapsto 0,
       C\setminus S\owns x \mapsto 1,
    \) and
    \(
       S\owns x  \mapsto 2.
    \)
(This is essentially $\fpair{\chi_S,\chi_C}: X\to 4$ without the impossible case
$x\in S\setminus C$.)
\end{definition}



\subparagraph*{Coalgebras.}
We briefly recall basic notions from coalgebra. For introductory
texts, see~\cite{Rutten00,JacobsR97,Adamek05,Jacobs17}.
Given an endofunctor $H: \C \to \C$, a \emph{coalgebra} is pair
$(C,c)$ where $C$ is an object of $\C$ called the \emph{carrier} and
thought of as an object of \emph{states}, and $c: C \to HC$ is a
morphism called the \emph{structure} of the coalgebra. Our leading examples are
the following.
\begin{example}\label{ex:coalg}
  \begin{enumerate}
  \item Labelled transition systems with labels from a set $A$ are
    coalgebras for the functor $HX = \Pot(A \times X)$ (and unlabelled
    transition systems are simply coalgebras for $\Pot$). Explicitly,
    a coalgebra $c:C\to HC$ assigns to each state $x$ a set
    $c(x)\in\Pot(A\times X)$, and this represents the transition
    structure at $x$: $x$ has an $a$-transition to $y$ iff
    $(a,y)\in c(x)$.
  \item Weighted transition systems with weights drawn from a
    commutative monoid are modelled as coalgebras as follows. For the
    given commutative monoid $(M,+,0)$, we consider the
    \emph{monoid-valued} functor $M^{(-)}$ on $\Set$ given for any map
    $h:X \to Y$ by
    \[
      M^{(X)} = \{ f: X \to M \mid f(x) \neq 0\text{ for finitely many
      }x\},
      \qquad
      M^{(h)}(f)(y) = \textstyle\sum_{hx = y} f(x).
    \]
    $M$-weighted transition systems are in bijective correspondence
    with coalgebras for $M^{(-)}$~\cite{gs01} (and for $M$-weighted
    labelled transition systems one takes $(M^{(-)})^A$).

  \item Probabilistic transition systems are modelled coalgebraically
    using the distribution functor~$\Dist$. This is the subfunctor
    $\Dist X\subseteq \R_{\ge 0}^{(X)}$, where $\R_{\ge 0}$ is the
    monoid of addition on the non-negative reals, given by
    $\Dist X=\{f\in\R_{\ge 0}^{(X)}\mid \sum_{x \in X} f(x) = 1\}$.

  \item\label{item:bags} The finite powerset functor $\Potf$ is a
    monoid-valued functor for the Boolean monoid $\B = (2, \vee, 0)$.
    The \emph{bag functor} $\Bagf$, which assigns to a set~$X$ the set
    of bags (i.e.\ finite multisets) on~$X$, is the monoid-valued
    functor for the additive monoid of natural
    numbers. 


  \item Simple (resp.~general) Segala systems~\cite{Segala95} strictly alternate
    between non-deterministic and probabilistic transitions; they can
    be modeled as coalgebras for the set functor $\Potf(A×\Dist(-))$ (resp.~$\Potf\Dist(A
    \times -)$).
  \end{enumerate}
\end{example}
A \emph{coalgebra morphism} from a coalgebra $(C,c)$ to a coalgebra
$(D,d)$ is a morphism $h: C \to D$ such that $d \cdot h = Hh \cdot c$;
intuitively, coalgebra morphisms preserve observable behaviour.
Coalgebras and their morphisms form a category $\Coalg(H)$. The
forgetful functor $\Coalg(H)\to \C$ creates all colimits, so
$\Coalg(H)$ has all colimits that $\C$ has.

A \emph{subcoalgebra} of a coalgebra $(C,c)$ is represented by a
coalgebra morphism $m: (D,d) \to (C,c)$ such that $m$ is a
monomorphism in $\C$. Likewise, a \emph{quotient} of a coalgebra
$(C,c)$ is represented by a coalgebra morphism
$q: (C,c) \to (D,d)$ carried by a regular epimorphism $q$ of $\C$. If
$H$ preserves monomorphisms, then the image factorization
structure on $\C$ lifts to coalgebras.  

\begin{definition}
  A coalgebra is \emph{simple} if it does not have any non-trivial
  quotients.
\end{definition}
Equivalently, a coalgebra $(C,c)$ is simple if every coalgebra
morphism with domain $(C,c)$ is carried by a monomorphism. Intuitively, in a simple
coalgebra all states exhibiting the same observable behaviour are
already identified. This paper is concerned with the design of
algorithms for computing \emph{the} simple quotient of a given coalgebra:
\begin{lemma}\label{lemmaSimple}
The simple quotient of a coalgebra is unique (up to isomorphism).
\end{lemma}
Intuitively speaking, two elements (possibly in different coalgebras)
are called behaviourally equivalent if they can be identified by
coalgebra morphisms. Hence, the simple quotient of a
coalgebra is its quotient modulo behavioural equivalence. In our main
examples, this means that we minimize w.r.t.\ standard
bisimilarity-type equivalences.
\begin{example}
  Behavioural equivalence instantiates to various notions of
  bisimilarity:
  \begin{enumerate}
  \item Park-Milner bisimilarity on labelled transition systems;
  \item weighted bisimilarity on weighted transition systems~\cite[Proposition~2]{Klin09};
  \item stochastic bisimilarity on probabilistic transition systems~\cite{Klin09};
  \item Segala bisimilarity on simple and general Segala systems~\cite[Theorem~4.2]{BARTELS200357}.

  \end{enumerate}
\end{example}

\section{A Categorical Algorithm for Behavioural Equivalence}
\label{sec:cat}

We proceed to describe a categorical partition refinement algorithm
that computes the simple quotient of a given coalgebra under fairly
general assumptions.

\begin{assumption}\label{ass:sec3}
  Assume that $H$ is an endofunctor on $\C$ that
  preserves monomorphisms.
\end{assumption}

\noindent Note that mono preservation is w.l.o.g.~for $\C=\Set$. Roughly,
for a given coalgebra $\xi:X\to HX$ in $\Set$, a \emph{partition
  refinement algorithm} maintains a quotient
$q: X\twoheadrightarrow \unnicefrac{X}{Q}$ that distinguishes some
(but possibly not all) states with different behaviour, and in fact,
initially~$q$ typically identifies everything. The algorithm repeats
the following steps:
\begin{enumerate}
\item Gather new information on which states should become separated
  by using $X\xrightarrow{\xi} HX \xrightarrow{Hq} H\unnicefrac{X}{Q}$,
  i.e., by identifying equivalence classes under $q$ that contain
  states whose behaviour is observed to differ under one more step of
  the transition structure $\xi$.
\item Use parts of this information to refine $q$ and repeat until $q$
  does not change any more.
\end{enumerate}
\noindent One of the core ideas of the Paige-Tarjan partition
refinement algorithm~\cite{PaigeTarjan87} is to not use all
information immediately in the second step. Recall that the algorithm
maintains two partitions $Y$ and $Z$ of the state set $X$ of the given
transition system; the elements of $Y$ are called \emph{subblocks} and
the elements of $Z$ are called \emph{compound blocks}. The partition
$Y$ is a refinement of the partition $Z$. The key to the
time efficiency of the algorithm is to select in each iteration a
subblock that is at most half of the size of the compound block it
belongs to. At the present high level of generality (which in
particular does not know about sizes of objects), we encapsulate the
subblock selection in a routine $\op{select}$, assumed as a parameter
to our algorithm:
\begin{definition}
  A \op{select} routine is an operation that receives a chain of two
  regular epis
  \(
    \begin{mytikzcd}
      |[inner sep=0mm]|
      X\  \arrow[->>]{r}{y}
      & Y \arrow[->>]{r}{z}
      & Z
    \end{mytikzcd}
  \)
  and returns some morphism $k: Y \to K$ into some object $K$. We call~$Y$ the
\emph{subblocks} and $Z$ the \emph{compound blocks}.
\end{definition}
The idea is that the morphism $k$ throws away some of the information
provided by the refinement $Y$. For example, in the Paige-Tarjan
algorithm it models the selection of one compound block to be split in
two parts, which then induce the further refinement of $Y$.

\begin{example} \label{exampleSelects}
\begin{enumerate}
\item \label{exampleSelectsChi} In the classical Paige-Tarjan
  algorithm~\cite{PaigeTarjan87}, i.e., for $\C = \Set$, one wants to
  find a proper subblock that is at most half of the size of the
  compound block it sits in. So let $S\in Y$ such that
  $2\cdot |y^{-1}[\{S\}]| \le |(zy)^{-1}[\{z(S)\}]|$. Here, $z(S)$ is
  the compound block containing $S$. Then we let $\op{select}(z,y)$ be
  $k: Y\to 3$ given by $k(x) = 2$ if $x=S$, else $k(x) = 1$ if
  $z(x) = z(S)$, and $k(x) = 0$ otherwise; i.e.\
  $k = \chi_{\{S\}}^{[S]_z}$ (\autoref{D:chi}). If~$Y$ and~$Z$ are
  encoded as partitions of $X$, then $S$ and $C:=z(S)$ are subsets of
  $X$ and $k\cdot y = \chi_S^C$.  If there is no such $S\in Y$, then
  $z$ is bijective, i.e., there is no compound block from $Z$ that
  needs to be refined. In this case, $k$ does not matter and we simply
  put $k=\mathbin{!}: Y \to 1$.

\item One obvious choice for $k$ is to take the identity on $Y$, so
  that \emph{all} of the information present in $Y$ is used for
  further refinement. We will discuss this in \autoref{finalchain}.

\item Two other, trivial, choices are $k=\mathbin{!}: Y\to 1$ and
  $k=z$. Since both of these choices provide no extra information,
  this will leave the partitions unchanged, see \autoref{noprogress}.
\end{enumerate}
\end{example}
\noindent Given a \op{select} routine, the most general form of our
partition refinement works as follows.
\begin{algorithm} \label{catPT} Given a coalgebra $\xi: X\to HX$, we
  successively refine equivalence relations~$Q$ and $P$ on $X$,
  maintaining the invariant that $P$ is finer than $Q$. In each step,
  we take into account new information on the behaviour of states,
  represented by a map $q:X\to K$, and accumulate this information in
  a map $\bar q:X\to \bar K$. To facilitate the analysis, these
  variables are indexed over loop iterations in the
  description. Initial values are
  \begin{equation*}
    Q_0  = X \times X\qquad
    q_0  = \mathbin{!}: X \to 1 = K_0\qquad
    P_0 =  \ker(
    X \xrightarrow{\xi}
    HX \xrightarrow{H!} H1).
  \end{equation*}
  We then iterate the following steps while $P_i \neq Q_i$,
  for $i \ge 0$:
  \begin{enumerate}
  \item \(
    \unnicefrac{X}{P_i} \overset{k_{i+1}}{\twoheadrightarrow} K_{i+1}
    := \op{select}\big(\!
      X \twoheadrightarrow
      \unnicefrac{X}{P_i} \twoheadrightarrow
      \unnicefrac{X}{Q_i}
    \!\big)
    \), 
    using that  $\unnicefrac{X}{P_i}$ is finer than $\unnicefrac{X}{Q_i}$

  \item $q_{i+1} := \begin{mytikzcd}[inlinecd]
      X \arrow[->>]{r}{}
      & \unnicefrac{X}{P_i} \arrow[->]{r}{k_{i+1}}
      &[2mm] K_{i+1}
    \end{mytikzcd}$,$\quad$
    $\bar q_{i+1} := \fpair{\bar q_i,q_{i+1}}:
    \begin{mytikzcd}[inlinecd]
      X \arrow{r}
      & \bar K_i × K_{i+1}
    \end{mytikzcd}$

  \item\label{step:Q} $Q_{i+1} := \ker\bar q_{i+1}$
    $\quad(= \ker \fpair{\bar q_i,q_{i+1}}= \ker \bar q_i \cap \ker q_{i+1})$

  \item\label{step:P}
    $P_{i+1} := \ker\big(\!\begin{mytikzcd}[inlinecd] X \arrow{r}{\xi}
      & HX \arrow{rr}{H\bar q_{i+1}} && H\prod_{j\le i+1} K_j
    \end{mytikzcd}\!\big)$ 

    \end{enumerate}
    Upon termination, the algorithm returns
    $\unnicefrac{X}{P_i}=\unnicefrac{X}{Q_i}$ as the simple quotient of
    $(X,\xi)$.
\end{algorithm}
\begin{notation}
  For spans $R\rightrightarrows X$, we will denote the canonical
  quotient by $\kappa_R: X\twoheadrightarrow \unnicefrac{X}{R}$.
\end{notation}
We proceed to prove correctness, i.e.\ that the algorithm really does
return the simple quotient of $(X,\xi)$. We fix the notation in
Algorithm~\ref{catPT} throughout. Since $\bar q$ accumulates more
information in every step, it is clear that $P$ and $Q$ are really
being successively refined:
\\[1mm]
\begin{minipage}{.4\textwidth}
\begin{lemma}\label{lem:inc} \label{PfinerthanQ} 
  For every $i$, $P_{i+1}$ is finer than $P_i$,
  $Q_{i+1}$ is finer than $Q_i$, and $P_i$ is finer
  than $Q_{i+1}$.
\end{lemma}
\end{minipage}\hfill%
\begin{minipage}{.55\textwidth}
\vspace{-4mm}
\begin{equation} \label{partitionSequence}
\hspace{-4mm}
\begin{mytikzcd}[column sep=5mm, row sep = 4mm,baseline=-1mm]
Q_0
& Q_1 \arrow[>->]{l}
& Q_2 \arrow[>->]{l}
&[3mm] Q_{i+1}
    \arrow[>->,dotted]{l}
    \arrow[>-,shorten >=6mm]{l}
    \arrow[->,shorten <=6mm]{l}
& Q_{i+2} \arrow[>->]{l}
&
    \arrow[->,shorten <=1mm,dotted]{l}
    \arrow[->,shorten <=3mm]{l}
\\
& P_0
    \arrow[>->]{u}
& P_1 \arrow[>->]{l}
    \arrow[>->]{u}
& \mathrlap{\,\,\,P_i}\phantom{Q_{i+1}}
    \arrow[>->,dotted]{l}
    \arrow[>-,shorten >=6mm]{l}
    \arrow[->,shorten <=6mm]{l}
    \arrow[>->]{u}
& P_{i+1} \arrow[>->]{l}
    \arrow[>->]{u}
&
    \arrow[->,shorten <=1mm,dotted]{l}
    \arrow[->,shorten <=3mm]{l}
\end{mytikzcd}
\hspace{-10mm}
\end{equation}
\end{minipage}
\\[1mm]
If we suppress the termination on $P_i = Q_i$ for a moment,
then the algorithm thus computes equivalence relations refining each
other.
At each step, \op{select} decides which part of the information
present in $P_i$ but not in $Q_i$ should be used to refine $Q_i$ to
$Q_{i+1}$. 
\\[-2mm]
\begin{minipage}[c]{.5\textwidth}
\begin{proposition}\label{propQuot}
  There exist morphisms
  $\unnicefrac{\xi}{Q_i}: \unnicefrac{X}{P_i} \to H(\unnicefrac{X}{Q_i})$
  for $i\ge 0$ (necessarily unique) such that \eqref{eq:xiQuotient} commutes.
\end{proposition}
\end{minipage}\hfill
\raisebox{2mm}{
\begin{minipage}[c]{.4\textwidth}
\begin{equation}
    \begin{mytikzcd}[
            row sep=4mm,
            column sep=12mm,
        ]
        X
            \arrow[->>]{r}{\kappa_{P_i}}
            \arrow{d}[left]{\xi\ }
        & \unnicefrac{X}{P_i}
            \arrow[dash pattern=on 2pt off 1pt]{d}{\ \unnicefrac{\xi}{Q_i}}
        \\
        HX
            \arrow[->]{r}[yshift=1pt]{H\kappa_{Q_i}}
        & H (\unnicefrac{X}{Q_i})
    \end{mytikzcd}
    \label{eq:xiQuotient}
  \end{equation}
\end{minipage}}

\noindent Upon termination the morphism
$\unnicefrac{\xi}{Q_i}$ yields the structure of a quotient coalgebra of $\xi$:
\begin{corollary}
  If $P_i = Q_i$ then $\unnicefrac{X}{Q_i}$ carries a unique coalgebra
  structure forming a quotient of $\xi: X\to HX$.
\end{corollary}
\noindent This means intuitively that all states that are merged by
the algorithm are actually behaviourally equivalent. The following
property captures the converse:
\begin{lemma}
    \label{soundness}
    Let $h: (X,\xi)\to (D,d)$ be a quotient of $(X,\xi)$. Then
    $\ker h$ is finer than both~$P_i$ and $Q_i$, for all $i\ge 0$.
\end{lemma}

\begin{theorem}[(Correctness)] \label{correctness}
  If $P_i = Q_i$, then $\unnicefrac{\xi}{Q_i}\!: \unnicefrac{X}{Q_i}\to
  H\,\unnicefrac{X}{Q_i}$ is a simple coalgebra.
\end{theorem}

\begin{remark}
  Most classical partition refinement algorithms are parametrized by an initial
  partition $\kappa_{\mathcal{I}}: X\twoheadrightarrow
  \unnicefrac{X}{\mathcal{I}}$. We start with the trivial partition $!: X \to 1$
  because a non-trivial initial partition might split equivalent behaviours and
  then would invalidate \autoref{soundness}. To accomodate an initial partition
  $\unnicefrac{X}{\mathcal{I}}$ coalgebraically, replace $(X,\xi)$ with the
  coalgebra $\fpair{\xi,\kappa_\mathcal{I}}$ for the functor
  $H(-)×\unnicefrac{X}{\mathcal{I}}$ -- indeed, already $P_0$ will then be finer
  than $\mathcal{I}$.
\end{remark}
\noindent We look in more detail at two corner cases of the algorithm,
where the \op{select} routine retains all available information,
respectively none:

\begin{remark} \label{finalchain} 
  Recall that $H$ induces the \emph{final sequence}:
    \[
      1 \xleftarrow{!} H1 \xleftarrow{H!} H^21 \xleftarrow{H^2!} \cdots
      \xleftarrow{H^{i-1}!} H^i 1 \xleftarrow{H^i !} H^{i+1} 1
      \xleftarrow{H^{i+1} !} \cdots
    \]
    Every coalgebra $\xi: X\to HX$ then induces a \emph{canonical
      cone} $\xi^{(i)}: X\to H^i 1$ on the final sequence, defined
    inductively by $\xi^{(0)}=\mathbin{!}$,
    $\xi^{(i+1)} = H\xi^{(i)}\cdot\xi$.  The objects $H^n 1$ may be
    thought of as domains of $n$-step behaviour for $H$-coalgebras. If
    $\C=\Set$ and $X$ is finite, then states $x$ and $y$ are
    behaviourally equivalent iff $\xi^{(i)}(x) = \xi^{(i)}(y)$ for all
    $i < \omega$ \cite{Worrell05}.

    The vertical inclusions in \eqref{partitionSequence} reflect that
    only some and not necessarily all of the information present in
    the relation $P_i$ (resp.~the quotient $\unnicefrac{X}{P_i}$) is
    used for further refinement. If indeed everything is used, i.e.,
    we have $k_{i+1} := \id_{\unnicefrac{X}{P_i}}$, then these
    inclusions become isomorphisms and then our algorithm simply
    computes the kernels of the morphisms in the canonical cone,
    i.e.~$Q_i = \ker\xi^{(i)}$.
\end{remark}
\noindent That is, when \op{select} retains all available information,
then Algorithm~\ref{catPT} just becomes a standard final chain
algorithm (e.g.~\cite{KonigKupper14}). The other extreme is the following:
\begin{definition}
  We say that \op{select} \emph{discards all new information at $i+1$}
  if $k_{i+1}$ factors through the morphism
  $\unnicefrac{X}{P_i} \twoheadrightarrow \unnicefrac{X}{Q_i}$
  witnessing that $P_i$ is finer than $Q_i$, see \autoref{PfinerthanQ}.
\end{definition}
\begin{lemma} \label{noprogress} The algorithm fails to progress in
  the $i+1$-th iteration, i.e.\ $Q_{i+1}=Q_i$, iff \op{select}
  discards all new information at $i+1$.
\end{lemma}
\begin{corollary}
  If $\C$ is (concrete over) $\Set^\Sorts$ and \op{select} never discards
  all new information, then \autoref{catPT} terminates and computes
  the simple quotient of a given finite coalgebra.
\end{corollary}
Indeed, \autoref{propQuot} shows that we obtain a chain of
successively finer quotients of $X$, and by \autoref{noprogress} this
chain must finally converge (i.e.\ $P_i=Q_i$ will hold).

\section{Incremental Partition Refinement}
\label{sec:opti}

\noindent In the most generic version of the partition refinement
algorithm (Algorithm~\ref{catPT}), the partitions are recomputed from
scratch in every step: In Step~\ref{step:P} of the algorithm,
$P_{i+1}=\ker(H\fpair{\bar q_i,q_{i+1}}\cdot\xi)$ is computed from the
information $\bar q_i$ accumulated so far and the new information
$q_{i+1}$, but in general one cannot exploit that the kernel
of $\bar q_i$ has already been computed. We now present a refinement
of the algorithm in which the partitions are computed incrementally,
i.e.\ $P_{i+1}$ is computed from $P_i$ and $q_{i+1}$. This requires
the type functor $H$ to be \emph{zippable} (\autoref{D:zip}). The
algorithm will be further refined in the next section.

Note that in Step~\ref{step:Q}, \autoref{catPT} computes a kernel
\( Q_{i+1} = \ker \bar q_{i+1} = \ker \fpair{\bar q_i, q_{i+1}}.  \)
In general, the kernel of a pair $\fpair{a,b}: D\to A×B$ is an
intersection 
$\ker a \cap \ker b$.
Hence, the partition for such a kernel can be computed in two
steps:
\begin{enumerate*}
\item Compute $\unnicefrac{D}{\ker a}$.
\hfill
\item Refine every block in $\unnicefrac{D}{\ker a}$ with respect to
$b: D\to B$.
\end{enumerate*}
\autoref{catPT} can thus be implemented to keep track of the partition
$\unnicefrac{X}{Q_i}$ and then refine this partition by $q_{i+1}$ in
each iteration. 

However, the same trick cannot be applied immediately to the
computation of $\unnicefrac{X}{P_i}$, because of the functor $H$
inside the computation of the kernel:
\( P_{i+1} = \ker (H\fpair{\bar q_i, q_{i+1}}\cdot \xi) \).
In the following, we will provide sufficient conditions for $H$,
$a: D\to A$, $b: D\to B$ to satisfy
\[
    \ker H\fpair{a,b}
    = \ker \fpair{Ha,Hb}.
\]
As soon as this holds for $a=\bar q_i, b=q_{i+1}$, we can optimize the
algorithm by changing Step~\ref{step:P}~to%
\begin{equation}
    P_{i+1}' := \ker \fpair{H\bar q_i\cdot \xi, Hq_{i+1}\cdot \xi}.
    \label{kernelOptimization}
\end{equation}

\begin{definition}\label{D:zip}
    A functor $H$ is \emph{zippable}
    if the following morphism is a monomorphism:
    \[
        \op{unzip}_{H,A,B}:
        H(A+B) \xrightarrow{\fpair{H(A+!),H(!+B)}} H(A+1) × H (1+B)
    \]
\end{definition}

\noindent Intuitively, if $H$ is a functor on $\Set$, we think of
elements $t$ of $H(A+B)$ as shallow terms with variables from
$A+B$. Then zippability means that each $t$ is uniquely determined by
the two terms obtained by replacing $A$- and $B$-variables,
respectively, by some placeholder $\gap$, viz.~the element of $1$, as
in the examples in \autoref{figZippable}.
\begin{figure}
    \begin{subfigure}[b]{.25\textwidth}
        \(
        \begin{mytikzcd}[row sep = 0mm]
        a_1\,a_2\,b_1\,a_3\,b_2
        \arrow[shiftarr={xshift=18mm},mapsto]{d}[xshift=-4mm,pos=0.0,above]{\op{unzip}}
        \\
        \begin{array}{c}
        (a_1a_2\gap a_3 \gap, \\
        \phantom{(}\,\gap \gap\,b_1\!\gap b_2)
        \end{array}
        \end{mytikzcd}
        \)
        \caption{$(-)^*$ is zippable}
    \end{subfigure}%
\hfill%
\begin{subfigure}[b]{.25\textwidth}
        \(
        \begin{mytikzcd}[row sep = 0mm, ampersand replacement = \&]
        \{a_1,a_2,b_1\}
        \arrow[shiftarr={xshift=18mm},mapsto]{d}[xshift=-4mm,pos=0.0,above]{\op{unzip}}
        \\
        \begin{array}{r@{\,}l}
        (\{a_1,a_2,&\gap\}, \\
        \{\gap,&b_1\})
        \end{array}
        \end{mytikzcd}
        \)
        \caption{$\Potf$ is zippable}
    \end{subfigure}%
\hfill%
\begin{subfigure}[b]{.44\textwidth}
        \(
        \begin{mytikzcd}[row sep = 0mm,
                       column sep = -13mm,
                       ]
        |[inner xsep=0mm]|
        \begin{array}{@{}l@{}}
        \big\{\{a_1,b_1\},
        \{a_2,b_2\}\big\}
        \end{array}
        \arrow[start anchor={[xshift=-4mm]},
               rounded corners,
               pos=0.3,
               to path={ |- (\tikztotarget) \tikztonodes },
               mapsto]{dr}[left,]{\op{unzip}}
        &
        &
        |[inner xsep=0mm]|
        \begin{array}{@{}l@{}}
        \big\{\{a_1,b_2\},
        \{a_2,b_1\}\big\}
        \end{array}
        \arrow[start anchor={[xshift=4mm]},
               rounded corners,
               pos=0.3,
               to path={ |- (\tikztotarget) \tikztonodes },
               mapsto]{dl}[right,overlay]{\op{unzip}}
        \\
        &
        |[inner xsep = 1mm]|
        \begin{array}{@{}l@{}}
        (\big\{\{a_1,\gap\},\{a_2,\gap\}\big\}, \\
        \phantom{(}\big\{\{\gap, b_1\},\{\gap,b_2\}\big\}) \\
        \end{array}
        \end{mytikzcd}
        \)
        \caption{$\Potf\Potf$ is not zippable}
    \end{subfigure}
    \caption{Zippability of \Set-Functors for sets
    $A=\{a_1,a_2,a_3\}$, $B=\{b_1,b_2\}$}
    \vspace{-4mm}
    \label{figZippable}
\end{figure}

In the following, we work in the category $\C = \Set^{\Sorts}$ of
$\Sorts$-sorted sets. However, most proofs are category-theoretic to
clarify where sets are really needed and where the arguments are more
generic.

\begin{example}
\begin{enumerate}
\item Constant functors $X\mapsto A$ are zippable: $\op{unzip}$
  is the diagonal $A \to A \times A$.

\item The identity functor is zippable since
    \(
        A+B \xrightarrow{\fpair{A+!, !+B}}
        (A+1) × (1+B)
    \)
    is monic in $\Set^{\Sorts}$.

  \item From Lemma~\ref{lem:closure} it follows that every
    polynomial endofunctor is zippable.
\end{enumerate}
\end{example}
\begin{lemma}\label{lem:closure}
  Zippable endofunctors are closed under products, coproducts and
  subfunctors. 
\end{lemma}
%
%
\begin{lemma}\label{lem:additive}
  If $H$ has a componentwise monic natural transformation
  $H(X+Y) \rightarrowtail HX × HY$, then $H$ is zippable.
\end{lemma}

\begin{example}
\begin{enumerate}
\item\label{item:monoid-zippable} For every commutative monoid, the
  monoid-valued functor $M^{(-)}$ admits a natural isomorphism
  $M^{(X+Y)} \cong M^{(X)} \times M^{(Y)}$, and hence is zippable by
  Lemma~\ref{lem:additive}.
\item As special cases of monoid-valued functors we obtain
  that the finite powerset functor $\Potf$ and the bag functor $\Bagf$
  are zippable.
\item The distribution functor $\Dist$ (see Example~\ref{ex:coalg}) is
  a subfunctor of the monoid-valued functor $M^{(-)}$ for $M$ the
  additive monoid of real numbers, and hence is zippable by
  Item~\ref{item:monoid-zippable} and Lemma~\ref{lem:closure}.
\item The previous examples together with the closure properties in
  Lemma~\ref{lem:closure} show that a number of functors of interest
  are zippable, e.g.~$2×(-)^\Inputs$,
  $2×\Pot(-)^\Inputs$, $\Pot(\Inputs×(-))$, $2×\big((-)+1\big)^\Inputs$,
  and variants where $\Pot$ is replaced with $\Bagf$, $M^{(-)}$, or
  $\Dist$.
\end{enumerate}
\end{example}

\begin{example}\label{non-zippable}
  The finitary functor $\Potf\Potf$ fails to be zippable, as shown in
  \autoref{figZippable}. First, this shows that zippable functors are not closed
  under quotients, since any finitary functor is a quotient of a polynomial,
  hence zippable, functor (recall that a $\Set$-functor $F$ is finitary if
  $FX=\bigcup\{Fi[Y]\mid i:Y\rightarrowtail X \text{ and }Y\text{ finite}\}$).
  Secondly, this shows that zippable functors are not closed under composition.
  One can extend the counterexample to a coalgebra to show that the optimization
  is incorrect for $\Potf\Potf$ and $\op{select} = \chi_S^C$. We will remedy
  this later by making use of a second sort, i.e.~by working in $\Set^2$
  (\autoref{sortedCoalgebra}).
\end{example}

\noindent Additionally, we will need to enforce constraints on the
\op{select} routine to arrive at the desired
optimization~\eqref{kernelOptimization}. This is because in general,
$\ker H\fpair{a,b}$ differs from $\ker\fpair{Ha,Hb}$ even for~$H$
zippable; e.g.\ for $H=\Pot$ and for $\pi_1,\pi_2$ denoting binary
product projections, $\fpair{\Pot\pi_1,\Pot\pi_2}$ in general fails to
be injective although $\Pot\fpair{\pi_1,\pi_2}=\Pot\id=\id$.

The next example illustrates this issue, and a related one: One might
be tempted to implement splitting by a subblock $S$ by $q_i = \chi_S$.
While this approach is sufficient for systems with real-valued
weights~\cite{ValmariF10}, it may in general let
$\ker (H\fpair{\bar q_i, q_{i+1}}\cdot \xi)$ and
$\ker \fpair{H\bar q_i\cdot \xi, Hq_{i+1}\cdot \xi}$ differ even for
zippable $H$, thus rendering the algorithm incomplete:

\begin{example}\label{expl:no-respect}
  Consider the coalgebra $\xi: X\to HX$ for the zippable functor
  $H=\{\blacktriangle, \mdblksquare, \smblkcircle\} × \Potf(-)$
  illustrated in \autoref{figSplitStep} (essentially a Kripke
  model). The initial partition $\unnicefrac{X}{P_0}$ splits by shape
  and by $\Potf!$, i.e.~states with and without successors are split
  (\autoref{figPartitionQ0}). Now, suppose that $\op{select}$ returns
  $k_1 := \id_{\unnicefrac{X}{P_0}}$, i.e.\ retains all information
  (cf.\ Remark~\ref{finalchain}), so that $Q_1 = P_0$ and $P_1$ puts
  $c_1$ and $c_2$ into different blocks (\autoref{figPartitionQ1}). We
  now analyse the next partition that arises when we split w.r.t.\ the
  subblock $S=\{c_1\}$ but not w.r.t.~the rest $C\setminus S$ of the
  compound block $C=\{c_1,c_2\}$; in other words, we take
  $k_2 := \chi_{\{\{c_1\}\}}: \unnicefrac{X}{P_1}\to 2$, making
  $q_2 = \chi_{\{c_1\}}: X\to 2$.  Then,
  $H\fpair{\bar q_1,q_2}\cdot \xi$ splits $t_1$ from $t_2$, because
  $t_1$ has a successor $c_2$ with $\bar q_1(c_2) = \{c_1,c_2\}$ and $q_2(c_2)
  = 0$ whereas $t_2$ has no such successor. However,
  $t_1,t_2$ fail to be split by $\fpair{H\bar q_1, Hq_2}\cdot \xi$
  because their successors do not differ when looking at successor
  blocks in $\unnicefrac{X}{Q_1}$ and $\unnicefrac{X}{\ker \chi_S}$
  separately: both have $\{c_1,c_2\}$ and $\{c_3\}$ as successor
  blocks in $\unnicefrac{X}{Q_1}$ and $\{c_1\}$ and
  $X\setminus\{c_1\}$ as successors in
  $\unnicefrac{X}{\ker\chi_S}$. Formally:
    \begin{align*}
        H\bar q_1\cdot \xi (t_1) &= 
        (\id×\Potf\kappa_{P_0})\cdot \xi (t_1) = 
        (\blacktriangle,\big\{
            \{c_1,c_2\},
            \{c_3\}
        \big\})
        = H\bar q_1\cdot \xi (t_2)
    \\
        H q_2\cdot \xi (t_1) &= 
        (\id× \Potf\chi_{\{c_1\}})\cdot \xi (t_1) = 
        (\blacktriangle,\big\{
            0, 1
        \big\})
        = H q_2\cdot \xi (t_2)
    \end{align*}
    So if we computed $P_2$ iteratively as in
    \eqref{kernelOptimization} for $q_2=\chi_S$, then $t_1$ and $t_2$
    would not be split, and we would reach the termination condition
    $P_2 = Q_2$ before all behaviourally inequivalent states have been
    separated.

    Already Paige and Tarjan \cite[Step~6 of the
    Algorithm]{PaigeTarjan87} note that one additionally needs to
    split by $C\setminus S=\{c_3\}$, which is accomplished by
    splitting by $q_i=\chi_S^C$. This is formally captured by
    the condition we introduce next.

\begin{figure}
    \tikzstyle{coalgebraNodes}=[
            every node/.style={
                draw=none,
                inner sep = 1pt,
                label distance=0mm,
            },
            every label/.append style={
                font=\small,
                execute at begin node=\(,
                execute at end node=\),
                inner sep=1pt,
                fill=white,
            },
    ]
    \tikzstyle{partitionPi}=[
        on background layer,
        every node/.style={
            shape=rectangle,
            rounded corners=2.5mm,
            minimum height=5mm,
            minimum width=5mm,
            draw=lipicsYellow,
            line width=1pt,
            inner sep = 1mm,
        },
    ]
    \tikzstyle{partitionQi}=[
        partitionPi,
        every node/.append style={
            dashed,
            inner sep = 2mm,
        },
    ]
    \begin{minipage}[b]{.48\textwidth}
    \newcommand{\exampleCoalgebra}{
    \begin{scope}[coalgebraNodes]
        \node[label={north:t_1}] (triangle1) at (2,1) {$\blacktriangle$};
        \node[label={north:t_2}] (triangle2) at (3,1) {$\blacktriangle$};
        \node[label={north:s_1}] (square1) at (1,1) {$\mdblksquare$};
        \node[label={south:c_1}] (circle1) at (1,0) {$\smblkcircle$};
        \node[label={south:c_2}] (circle2) at (2,0) {$\smblkcircle$};
        \node[label={south:c_3}] (circle3) at (3,0) {$\smblkcircle$};
    \end{scope}
    \path[->]
        (triangle1) edge (circle1)
        (triangle1) edge (circle2)
        (triangle1) edge (circle3)
        (triangle2) edge (circle1)
        (triangle2) edge (circle3)
        (circle1) edge (square1)
        (circle2) edge (circle3)
        ;
    }
    \begin{subfigure}{.48\textwidth}
    \begin{tikzpicture}
        \exampleCoalgebra
        \begin{scope}[partitionQi ]
            \node[fit = (square1) (circle1) (circle2) (circle3)
                        (triangle2) (triangle1)] {};
        \end{scope}
        \begin{scope}[partitionPi ]
            \node[fit = (square1)] {};
            \node[fit = (circle1) (circle2)] {};
            \node[fit = (circle3)] {};
            \node[fit = (triangle2) (triangle1)] {};
        \end{scope}
    \end{tikzpicture}
    \caption{$Q_0, P_0$ for $\bar q_0 = \mathbin{!}$}
    \label{figPartitionQ0}
    \end{subfigure}%
\hfill\begin{subfigure}{.49\textwidth}
    \begin{tikzpicture}
        \exampleCoalgebra
        \begin{scope}[ partitionQi ]
            \node[fit = (square1)] {};
            \node[fit = (circle1) (circle2)] {};
            \node[fit = (circle3)] {};
            \node[fit = (triangle2) (triangle1)] {};
        \end{scope}
        \begin{scope}[ partitionPi ]
            \node[fit = (square1)] {};
            \node[fit = (circle2)] {};
            \node[fit = (circle1)] {};
            \node[fit = (circle3)] {};
            \node[fit = (triangle2) (triangle1)] {};
        \end{scope}
    \end{tikzpicture}
    \caption{$Q_1, P_1$ for $\bar q_1 =
    \kappa_{P_0}$}\noshowkeys\label{figPartitionQ1}
    \end{subfigure}
    \caption{Partitions of a coalgebra $\xi$ for $H=\{\blacktriangle,
    \mdblksquare, \smblkcircle\} × \Potf(-)$. $\unnicefrac{X}{Q_i}$ is
    indicated by dashed, $\unnicefrac{X}{P_i}$ by solid lines.}
    \label{figSplitStep}
    \end{minipage}%
\hfill\begin{minipage}[b]{.45\textwidth}
    \begin{tikzpicture}[
            elementX/.style={
                execute at begin node=\(,
                execute at end node=\),
            },
            x=7mm,
            y=8mm,
            baseline=-5mm, 
        ]
        \foreach \prefix/\yshift in {top/2,middle/1,bottom/0} {
        \node[elementX] (\prefix circle1) at (0,\yshift) {c_1};
        \node[elementX] (\prefix circle2) at (1,\yshift) {c_2};
        \node[elementX] (\prefix circle3) at (2,\yshift) {c_3};
        \node[elementX] (\prefix square1) at (3,\yshift) {s_1};
        \node[elementX] (\prefix triangle1) at (4,\yshift) {t_1};
        \node[elementX] (\prefix triangle2) at (5,\yshift) {t_2};
        }
        \begin{scope}[
                partitionPi,
                every node/.append style={inner sep=0mm},
                ]
            \node[fit = (topsquare1)] {};
            \node[fit = (topcircle1) (topcircle2)] {};
            \node[fit = (topcircle3)] {};
            \node[fit = (toptriangle2) (toptriangle1)] (XQ1) {};
        \end{scope}
        \begin{scope}[
                partitionPi,
                every node/.append style={inner sep=0mm},
                ]
            \node[fit =  (middlecircle1)] {};
            \node[fit = (middlesquare1)
                        (middlecircle2) (middlecircle3) (middletriangle2)
                        (middletriangle1)] (chiS) {};
        \end{scope}
        \begin{scope}[
                partitionPi,
                every node/.append style={inner sep=0mm},
                ]
            \node[fit =  (bottomcircle1)] {};
            \node[fit = (bottomcircle2) ] {};
            \node[fit = (bottomcircle3) (bottomsquare1) (bottomtriangle2)
                        (bottomtriangle1)]  (chiSC){};
        \end{scope}
        \begin{scope}[every node/.style={
                inner xsep=0cm,
                anchor=west,
                xshift=1mm
            }]
        \node at (chiS.east) {$\unnicefrac{X}{\ker \chi_S}$};
        \node at (chiSC.east) {$\unnicefrac{X}{\ker \chi_S^C}$};
        \node at (XQ1.east) {$\unnicefrac{X}{Q_1} = \unnicefrac{X}{P_0}$};
        \end{scope}
    \end{tikzpicture}
    \caption{Grouping of elements when $S:=\{c_1\}$ is chosen as the
      next subblock and $C := \{c_1,c_2\}$ as the compound block.  }
    \noshowkeys\label{groupingExamples}
    \end{minipage}
\end{figure}
\end{example}

\begin{definition} \label{defRespectCompounds}
  A \op{select} routine \emph{respects compound blocks} if whenever
  $k= \op{select}(X\overset{y}{\twoheadrightarrow} Y
  \overset{z}{\twoheadrightarrow} Z)$
  then the union $\ker k\cup \ker z$ is a kernel.
\end{definition}
In $\Set^{\Sorts}$, $\cup$ denotes the usual union of multi-sorted relations;
and since reflexive and symmetric relations are closed under unions,
the definition boils down to $\ker k\cup\ker z$ being transitive. We
can rephrase the condition more explicitly, restricting to the
single-sorted case for readability:

\begin{lemma} \label{compoundBlockEquivalences}
    For $a: D\to A$, $b: D\to B$ in \Set, the following are
    equivalent:
    \begin{enumerate}
    \item $\ker a\cup \ker b\rightrightarrows D$ is a kernel.
    \item $\ker a\cup \ker b\rightrightarrows D$ is the kernel of
    the pushout of $a$ and $b$.
    \item For all $x,y,z \in D$, $ax =ay$ and $by=bz$ implies
    $ax = ay = az$ or $bx=by=bz$.
    \item For all $x\in D$, $[x]_a \subseteq [x]_b$ or  $[x]_b
    \subseteq [x]_a$.
    \label{classInclusion}
    \end{enumerate}
\end{lemma}
The last item states that when going from $a$-equivalence classes to
$b$-equivalence classes, the classes either merge or split, but do not
merge with other classes and split at the same time. Note that in
\autoref{groupingExamples}, $Q_1\cup \ker \chi_S$ fails to be
transitive, 
while $Q_1\cup \ker \chi_S^C$ is transitive.

\begin{example} \label{examplesRespectCompoundBlocks}
All $\op{select}\big(
X\!\overset{y}{\twoheadrightarrow}\!Y\!
    \overset{z}{\twoheadrightarrow}\!Z
\big)
$ routines in \autoref{exampleSelects} respect compound blocks.
\end{example}

\begin{proposition} \label{propZippable} Let $a: D\to A$, $b: D\to B$
  be a span such that $\ker a\cup \ker b$ is a kernel, and let
  $H: \Set\to \D$ be a zippable functor preserving monos. Then we have
    \begin{equation}
        \ker \fpair{Ha,Hb} = \ker H\fpair{a,b}.
        \label{eqKerFpair}
    \end{equation}
\end{proposition}
\noindent We thus obtain soundness of optimization
\eqref{kernelOptimization}; summing up:
\begin{corollary}\label{optimizationSummary}
  Suppose that $H$ is a zippable endofunctor on $\Set$ and that
  $\op{select}$ respects compound blocks and never discards all new
  information. Then \autoref{catPT} with
  optimization~\eqref{kernelOptimization} terminates and computes the
  simple quotient of a given finite $H$-coalgebra.
\end{corollary}

\begin{remark} \label{sortedCoalgebra} Like most results on set
  coalgebras, the above extends to multisorted sets by componentwise
  arguments, and this allows dealing with complex composite
  functors~\cite{SchroderPattinson11}. We restrict to a case with
  lightweight notation: Let $F$ and $G$ be zippable \Set-functors,
  recalling from \autoref{non-zippable} that the composite $FG$ need
  not itself be zippable, and let $F$ be finitary.
  Then in lieu of $FG$-coalgebras, we can equivalently consider
  coalgebras for the endofunctor \( H:(X,Y) \mapsto (FY,GX) \)
  on $\Set^2$. In particular, a coalgebra $\xi: X\to FGX$ with finite
  carrier~$X$ induces, since~$F$ is finitary, a finite set
  $Y\subseteq_\mathrm{f} GX$, with inclusion~$y$, such that
  $\xi = (X \xrightarrow{x} FY \xrightarrow{Fy} FGX)$, so we obtain a
  finite $H$-coalgebra $(x,y): (X,Y) \to (FY,GX)$.  Mutatis mutandis,
  \autoref{propZippable} holds also for $H$, since kernels and pairs
  in $\Set^2$ are computed componentwise, so we obtain a version of
  \autoref{optimizationSummary} for~$H$. Explicitly, when computing
  the kernel of
  $H\bar q_{i+1} \cdot (x, y)$, we can use
  the optimization \eqref{kernelOptimization} in both sorts. The first
  component of the simple quotient of $((X,Y),(x,y))$ computed by the
  algorithm then yields the simple quotient of the
  original~$(X,\xi)$. Composites of more than two functors are treated
  similarly.
\end{remark}
\begin{example}\label{segala-twosorted}
  Applying this to the functors $F = \Potf$, $G=A\times(-)$, and
  $H = \Dist$, we obtain simple (resp.~general) Segala systems as
  coalgebras for $FGH$ (resp.~$FHG$). For simple Segala systems, the
  $\Set^3$ functor is defined by $(X,Y,Z) \mapsto (FY,GZ,HX)$.
\end{example}

\section{Efficient Calculation of Kernels}
\label{sec:efficient}
In \autoref{catPT}, it is left unspecified how the kernels are
computed concretely. We proceed to define a more concrete algorithm
based on a \emph{refinement interface} of the functor. This interface
is aimed at efficient implementation of the \emph{refinement step} in
the algorithm. Specifically from now on, we split
along $\xi\!: X\to HY$
w.r.t.\ a subblock $S\subseteq C\in\unnicefrac{Y}{Q}$, and need to
compute how the splitting of $C$ into $S$ and $C\setminus S$ within
$\unnicefrac{Y}{Q}$ affects the partition $\unnicefrac{X}{P}$.

The low complexity of Paige-Tarjan-style algorithms hinges on this
refinement step running in time $\CO(|\op{pred}[S]|)$, where
$\op{pred}(y)$ denotes the set of predecessors of some $y\in Y$ in the
given transition system. In order to speak about ``predecessors''
w.r.t.~more general $\xi: X\to HY$, the refinement interface will
provide an encoding of $H$-coalgebras as sets of states with successor
states encoded as bags (implemented as lists up to ordering; recall
that $\Bagf Z$ denotes the set of bags over $Z$) of
$A$-labelled edges, where $A$ is an appropriate label alphabet. Moreover, the
interface will allow us to talk about the behaviour of elements of $X$
w.r.t.~the splitting of $C$ into $S$ and $C\setminus S$, looking only
at points in $S$.

\begin{definition} \label{refinementInterface}
    A \emph{refinement interface} for a $\Set$-functor $H$ is formed
    by a set $A$ of \emph{labels}, a set $W$ of \emph{weights} and functions
    \begin{align*}
    \begin{array}{l@{\qquad}l}
        \flat: HY \to \Bagf(A×Y),& \op{init}: H1×\Bagf A\to W,\\
        w: \Potf Y \to HY \to W,& \op{update}: \Bagf A × W \to W× H3× W
    \end{array}
    \end{align*}
    such that for every
    $S\subseteq C\subseteq Y$, the diagrams
    \vspace{-2mm}
    \begin{equation}
    \label{eqSplitterLabels}
    \hspace{-4mm}
    \begin{mytikzcd}
        HY
        \arrow{d}[left, pos=0.4]{\begin{array}{r}
            \fpair{
            H!,
            \\[-1mm]
            \Bagf\pi_1\cdot \flat}
            \end{array}}
        \arrow[shift left=2]{dr}{w(Y)}
        \\
        H1 × \Bagf A
            \arrow{r}{\op{init}}
        & W
    \end{mytikzcd}
    \hspace{1mm}
    \begin{mytikzcd}[column sep=1.4cm]
    HY
    \arrow[
        to path={
            -| (\tikztotarget) \tikztonodes
        },
        rounded corners,
    ]{drr}[pos=0.23,below,yshift=-1mm]{\fpair{w(S),H\chi_S^C, w(C\setminus S)}}
    \arrow{d}[left]{\fpair{\flat,w(C)}}
    \\
    \Bagf(A×Y) × W
    \arrow{r}{\op{fil}_S × W}
    & \Bagf A × W
    \arrow{r}{\op{update}}
    & W × H3 × W
    \end{mytikzcd}
    \hspace{-3mm}
    \end{equation}
\end{definition}
\noindent
commute, where \( \op{fil}_S: \Bagf(A×Y) \to \Bagf(A)\)
is the filter function
\(\op{fil}_S (f)(a) = \textstyle\sum_{y \in S} f(a,y)\)
for $S\subseteq Y$.  The significance of the set $H3$ is that when
using a set $S\subseteq C\subseteq X$ as a splitter, we want to split
every block $B$ in such a way that it becomes \emph{compatible} with
$S$ and $C\setminus S$, i.e.~we group the elements $s\in B$ by the
value of $H\chi_{S}^C\cdot\xi (s) \in H3$. The set $W$ depends on the
functor. But in most cases $W=H2$ and $w(C) = H \chi_C: HY \to H2$ are
sufficient.

In an implementation, we do not require a refinement interface to
provide $w$ explicitly, because the algorithm will compute the values
of $w$ incrementally using \eqref{eqSplitterLabels}, and $\flat$ need
not be implemented because we assume the input coalgebra to be already
\emph{encoded} via $\flat$:
\begin{definition}
  Given an interface of $H$ (\autoref{refinementInterface}),
  an \emph{encoding} of a morphism $\xi: X\to HY$ is given by a set
  $E$ and maps
    \begin{align*}
    \op{graph}&: E \to X×A×Y &
    \op{type}&: X \to H1
    \end{align*}
    such that $(\flat\cdot \xi(x))(a,y) = |\{e\in E\mid \op{graph}(e) =
    (x,a,y)\}|$, and with $\op{type} = H!\cdot
    \xi$.
\end{definition}
\noindent 
Intuitively, an encoding presents the morphism $\xi$ as a graph with
edge labels from $A$.

\begin{lemma} \label{encodingCanonical}
  Every morphism $\xi: X\to HY$ has a canonical encoding where $E$ is
  the obvious set of edges of $\flat\cdot\xi: X\to \Bagf(A×Y)$.
  If $X$ is finite, then so is $E$.
\end{lemma}

\begin{example} \label{exampleH2} In the following examples, we take
  $W = H2$ and $w(C) = H\chi_C: HY\to H2$. We use the helper function
  $\op{val} := \fpair{H(=2), \id, H(=1)}: H3\to H2×H3×H2$, where
  $(=x): 3\to 2$ is the equality check for $x\in \{1,2\}$, and in each
  case define $\op{update} = \op{val}\cdot \op{up}$ for some function
  $\op{up}: \Bagf A× H2\to H3$. We implicitly convert sets into bags.
\begin{enumerate}
\item\label{exampleH2.1} For the monoid-valued functor $G^{(-)}\!$,
  for an Abelian group $(G,+,0)$, we take labels $A = G$ and define
  $\flat(f) = \{ (f(y),y) \mid y\in Y, f(y) \neq 0\}$ (which is is
  finite because $f$ is finitely supported). With $W= H2= G×G$, the
  weight $w(C) = H\chi_C: HY\to G×G$ is the accumulated weight of
  $Y\setminus C$ and $C$. Then the remaining functions are
\begin{align*}
\op{init}(h_1, e) = (0, \groupsum e) 
\quad
\text{and}
\quad
\op{up}(e, (r,c)) =
    (r, c - \groupsum e,  \groupsum e),
\end{align*}
where $\groupsum: \Bagf G \to G$ is the obvious summation map. 

\item Similarly to the case $\R^{(-)}$, one has the following
  $\op{init}$ and $\op{up}$ functions for the distribution functor
  $\Dist$: put $\op{init}(h_1, e) = (0,1) \in \Dist 2 \subset [0,1]^2$
  and $\op{up}(e, (r,c)) = (r, c - \groupsum e, \groupsum e)$ if the
  latter lies in $\Dist 3$, and $(0,0,1)$ otherwise.

\item Similarly, one obtains a refinement interface for $\Bagf =
  \N^{(-)}$ adjusting the one for $\Z^{(-)}\!$; in fact,
  $\op{init}$ remains unchanged and $\op{up}(e, (r,c)) = (r,
  c-\groupsum e, \groupsum e)$ if the middle component is a natural
  number and $(0,0,0)$ otherwise. 

\item Given a polynomial functor $H_\Sigma$
  for a signature $\Sigma$ with bounded arity (i.e.~there exists $k$
  such that every arity is at most $k$),
  the labels $A = \N$ encode the indices of the parameters:
\begin{align*}
&\flat(\sigma(y_1,\ldots,y_n)) = \{ (1,y_1),\ldots,(n,y_n) \}\quad
\op{init}(\sigma(0,\ldots,0), f) = \sigma(1,\ldots,1)
\\
&\op{up}(I, \sigma(b_1,\ldots,b_n))
   = \sigma(b_1 + (1\in I),\ldots,b_i + (i\in I),\ldots,b_n + (n\in I))
\end{align*}
Here $b_i + (i\in I)$ means $b_i + 1$ if $i\in I$ and $b_i$ otherwise.
Since $I$ are the indices of the parameters in the subblock, $i\in I$
happens only if $b_i = 1$.
\end{enumerate}
\end{example}
One example where $W=H2$
does not suffice is the powerset functor $\Pot$:
Even if we know for a $t\in
\Pot Y$ that it contains elements in $C \subseteq Y$, in $S\subseteq
C$, and outside $C$ (i.e.~we know $\Pot{\chi_S}(t), \Pot{\chi_C} \in
\Pot
2$), we cannot determine whether there are any elements in $C\setminus
S$ -- but as seen in \autoref{expl:no-respect}, we need to include
this information.
\begin{example} \label{examplePowerset} The interface for the powerset
  functor needs to count the edges into blocks $C
  \supseteq
  S$ in order to know whether there are edges into $C\setminus
  S$, as described by Paige and Tarjan~\cite{PaigeTarjan87}.  What
  happens formally is that first the interface for $\N^{(-)}$
  is implemented for edge weights at most $1$,
  and then the middle component of the result of $\op{update}$
  is adjusted. So $W
  = \N^{(2)}$, and the encoding $\flat: \Potf Y\hookrightarrow
  \Bagf(1×Y)$ is the obvious inclusion. Then
\begin{align*}
    \op{init}(h_1, e) &= (0,|e|)
    \quad
    \quad
    w(C)(M) = \Bagf{\chi_C}(M) = (|M\setminus C|, |C\cap M|)
    \\
    \op{update}(n, (c, r))
        &= \fpair{\Bagf{(=2)}\!,\ (\geZero 0)^3\!,\ \Bagf{(=1)}}
            (r, c - |n|, |n|)
        \\
        &= \big(
            (r + c - |n|,|n|),
            (r \geZero  0,
             c - |n| \geZero  0,
             |n| \geZero  0),
            ( r + |n|, c - |n|)
        \big),
\end{align*}
where $x \geZero  0$ is $0$ if $x = 0$ and $1$ otherwise. 
\end{example}
\begin{assumption}\label{interface-linear}
  From now on, assume a $\Set$-functor
  $H$
  with a refinement interface such that $\op{init}$
  and $\op{update}$
  run in linear time and elements of $H3$
  can be compared in constant time.
\end{assumption}
\begin{example}
  \label{exampleTime}
  The refinement interfaces in Examples~\ref{exampleH2}
  and~\ref{examplePowerset} satisfy \autoref{interface-linear}.
\end{example}
\begin{remark}
  In the implementation, we encode the partitions
  $\unnicefrac{X}{P}$,
  $\unnicefrac{Y}{Q}$
  as doubly linked lists of the blocks they contain, and each block is
  in turn encoded as a doubly linked list of its elements. The
  elements $x\in
  X$ and $y\in
  Y$ each hold a pointer to the corresponding list entry in the blocks
  containing them. This allows removing elements from a block in
  $\CO(1)$.
%
\end{remark}
\noindent The algorithm maintains the following mutable
data structures:
\begin{itemize}
\item An array $\op{toSub}:
  X\to \Bagf E$, mapping $x\in
  X$ to its outgoing edges ending in the currently processed subblock.
\item A pointer mapping edges to memory
addresses: $\op{lastW}: E \to \N$.
\item A store of last values $\op{deref}: \N\to W$.
\item For each block $B$ a list of markings $\op{mark}_B \subseteq
B×\N$.
\end{itemize}
\begin{notation}
    In the following we write $e = x\xrightarrow{a}y$ in lieu of $\op{graph}(e) =
    (x,a,y)$.
\end{notation}

\begin{definition}[(Invariants)]
Our correctness proof below establishes the following properties that we call \emph{the invariants}: 
\begin{enumerate}
\item \label{invariant_toSub}
    For all $x\in X, \op{toSub}(x) = \emptyset$, i.e.~$\op{toSub}$ is
    empty everywhere.

\item \label{invariant_lastW}
    For $e_i = x_i\xrightarrow{a_i}y_i, i\in \{1,2\}$,
\(
    \op{lastW}(e_1)
    = \op{lastW}(e_2)
\ \Longleftrightarrow\ %
    x_1=x_2 \text{ and }[y_1]_{\kappa_Q} = [y_2]_{\kappa_Q}.
\)

\item\label{invariant_deref} For each $e = x\xrightarrow{a} y$, $C :=
[y]_{\kappa_Q} \in \unnicefrac{Y}{Q}$,
\(
    w(C, \xi(x)) = \op{deref}\cdot \op{lastW}(e).
\)

\item \label{invariant_Hchi}
    For $x_1,x_2\in B\in \unnicefrac{X}{P}$, $C\in \unnicefrac{Y}{Q}$,
    $(x_1,x_2) \in \ker (H\chi_C\cdot \xi)$.

\end{enumerate}
\end{definition}
\noindent In the following code listings, we use square brackets for
array lookups and updates in order to emphasize they run in constant
time. We assume that the functions
$\op{graph}: E \to X \times A \times Y$ and $\op{type}: X \to H1$ are
implemented as arrays.  In the initialization step the predecessor
array $\op{pred}: Y\to \Potf E$,
\( \op{pred}(y) = \{ e\in E \mid e = x\xrightarrow{a} y \} \)
is computed.  Sets and bags are implemented as lists. We only insert
elements into sets not yet containing them.

We say that we \emph{group} a finite set $Z$ by $f: Z \to Z'$ to
indicate that we compute $[-]_f$. This is done by sorting the elements
of $z\in Z$ by a binary encoding of $f(z)$ using any
$\CO(|Z|\cdot \log|Z|)$ sorting algorithm, and then grouping elements
with the same $f(z)$ into blocks. In order to keep the overall
complexity for the grouping operations low enough, one needs to use a
possible majority candidate during sorting, following Valmari and
Franceschinis~\cite{ValmariF10}. 
\begin{figure}[!h]
\begin{algorithmic}[1]
    \For{$e\in E$, $e = x\xrightarrow{a} y$}
        \State add $e$ to $\op{toSub}[x]$ and $\op{pred}[y]$.
    \EndFor
    \For{$x\in X$}
        \State $p_X \gets$ new cell in $\op{deref}$ containing
        $\op{init}(\op{type}[x],
        \Bagf(\pi_2\cdot\op{graph})(\op{toSub}[x]))$
        \State \textbf{for} $e\in \op{toSub}[x]$ \textbf{do}
        \label{algoLineInit} $\op{lastW}[e] = p_X$
        \State $\op{toSub}[x] := \emptyset$
        \label{algoLineInitEmptyset}
    \EndFor
    \State $\unnicefrac{X}{P} \gets $ group $X$ by $\op{type}: X \to H1$,
    $\unnicefrac{Y}{Q} \gets \{ Y\}$.
\end{algorithmic}
\vspace{-1mm}
\caption{The initialization procedure}
\vspace{-3mm}
\label{figAlgoInit}
\end{figure}
The algorithm computing the initial partition is listed in
\autoref{figAlgoInit}.
\begin{lemma}\label{lemmaInit}
The initialization procedure runs in time $\CO(|E|+|X|\cdot\log|X|)$ and
makes the invariants true.
\end{lemma}
\vspace{-1mm}
\noindent The algorithm for one refinement step along a morphism
$\xi: X\to HY$ is listed in \autoref{figSplitAlgo}.
\begin{figure}[b!]
\begin{subfigure}[t]{.45\textwidth}
\textsc{Split}$(\unnicefrac{X}{P}, \unnicefrac{Y}{Q}, S\subseteq C\in \unnicefrac{Y}{Q})$
\begin{algorithmic}[1]
\State{$\op{M} \gets \emptyset \subseteq \unnicefrac{X}{P}×H3$}
\noshowkeys\label{algoFirstLine}
\For{$y\in S, e\in \op{pred}[y]$}
    \State $x\xrightarrow{a}y \gets e$
    \State $B \gets$ block with $x\in B\in \unnicefrac{X}{P}$
    \If{$\op{mark}_B$ is empty}
        \State $w_C^x \gets \op{deref}\cdot\op{lastW}[e]$
        \State $v_\emptyset \gets
        \pi_2\cdot\op{update}(\emptyset,w_C^x)$
        \State add $(B,v_\emptyset)$ to
        $\op{M}$
    \EndIf
    \If{$\op{toSub}[x] = \emptyset$}
        \State add $(x,\op{lastW}[e])$ to $\op{mark}_B$\label{algo10}
    \EndIf
    \State add $e$ to $\op{toSub}[x]$
\EndFor
\algstore{splitAlgo}
\end{algorithmic}
\end{subfigure}\hfill%
\begin{subfigure}[t]{.52\textwidth}
\begin{algorithmic}[1]
\algrestore{splitAlgo}
\For{$(B,v_\emptyset) \in \op{M}$}
    \noshowkeys\label{algoSplitLoop}
    \State $B_{\neq\emptyset} \gets \emptyset \subseteq X × H3$
    \For{$(x,p_{C})$ in $\op{mark}_B$}
        \noshowkeys\label{algoLineMarkedX}
        \State $\ell \gets \Bagf(\pi_2\cdot \op{graph})(\op{toSub}[x])$
            \noshowkeys\label{algoLineLabels}
        \State $(w_S^x, v^x,w_{C\setminus S}^x) \gets
            \op{update}(\ell, \op{deref}[p_{C}]\mathrlap{)}$%
\noshowkeys\label{algoLineUpdate}
        \State $\op{deref}[p_{C}] \gets w^x_{C\setminus S}$
        \State $p_S\gets $ new cell containing $w_S^x$
            \noshowkeys\label{algoLineNewCell}
        \State \textbf{for} $e \in \op{toSub}[x]$ \textbf{do} $\op{lastW}[e] \gets p_S$
            \noshowkeys\label{algoLineSetLastW}
        \State $\op{toSub}[x] \gets \emptyset$
        \If{$v^x \neq v_\emptyset$} \noshowkeys\label{algoLineIfVXVempty}
            \State remove $x$ from $B$
            \State insert $(x,v^x)$ into $B_{\neq\emptyset}$
            \noshowkeys\label{algoLineInsertVX}
        \EndIf
    \EndFor
    \label{algoLastLine}
    \State \label{algoLineSplit}
        \(
        \begin{array}[t]{l}
        B_1×\{v_1\},\ldots,B_\ell×\{v_\ell\} \gets
            \\
            \quad\text{group }B_{\neq\emptyset}\text{ by }\pi_2:X×H3\to H3
        \end{array}
        \)
    \State insert $B_1,\ldots,B_\ell \gets$ into $\unnicefrac{X}{P}$
\EndFor
\end{algorithmic}
\end{subfigure} \\[2mm]
\begin{subfigure}[t]{.45\textwidth}
    \caption{Collecting predecessor blocks}
    \label{algoPredCollecting}
\end{subfigure}\hfill%
\begin{subfigure}[t]{.52\textwidth}
    \caption{Splitting predecessor blocks} \noshowkeys\label{algoPredSplitting}
\end{subfigure}
\caption{Refining $\unnicefrac{X}{P}$ w.r.t~$
\chi_S^C: Y\to 3$ and
$\unnicefrac{Y}{Q}$ along $\xi: X\to HY$}
\label{figSplitAlgo}
\end{figure}
\noindent In the first part, all blocks $B\in \unnicefrac{X}{P}$ are
collected that have an edge into $S$, together with
$v_\emptyset\in H3$ which represents $H\chi_S^C\cdot\xi(x)$ for any
$x\in B$ that has no edge into $S$. For each $x\in X$, $\op{toSub}[x]$
collects the edges from $x$ into $S$. The markings $\op{mark}_B$ list
those elements $x\in B$ that have an edge into $S$, together with a
pointer to $w(C,x)$.

In the second part, each block $B$ with an edge into $S$
is refined w.r.t.~$H\chi_S^C\cdot \xi$. First, for any $(x,p_C)\in
\op{mark}_B$, we compute $w(S,x)$, $v^x = H\chi_S^C\cdot\xi(x)$, and
$w(C\setminus S,x)$ using $\op{update}$. Then, the weight of all edges
$x \to C\setminus S$ is updated to $w(C\setminus S, x)$ and the
weight of all edges $x\to S$ needs to be stored in a new cell
containing $w(S,x)$. For all unmarked $x\in B$, we
know that $H\chi_S^C\cdot\xi(x) = v_\emptyset$; so all $x$ with
$v^x=v_\emptyset$
stay in $B$. All other $x\in B$ are removed and distributed to new
blocks w.r.t.~$v^x$.

\begin{theorem} \label{thmSplitCorrect}
    \textsc{Split}$(\unnicefrac{X}{P}, \unnicefrac{Y}{Q} ,S\subseteq C\in \unnicefrac{Y}{Q})$ refines
    $\unnicefrac{X}{P}$ by $\smash{H\chi_S^C\cdot \xi: X\to H3}$.
\vspace{-1pt}
\end{theorem}
\begin{lemma}\label{lemmaAfterSplit}
    After running $\textsc{Split}$, the invariants hold.
\vspace{-1pt}
\end{lemma}
\begin{lemma}\label{lemmaTimeSplit}
  Lines $\ref{algoFirstLine}-\ref{algoLastLine}$ in \textsc{Split} run
  in time $\smash{\CO(\sum_{y\in S}|\op{pred}(y)|)}$.
\vspace{-1pt}
\end{lemma}
\begin{lemma}\label{lemmaTime}
    For $S_i\subseteq C_i \in \unnicefrac{Y}{Q_i}, 0\le i< k$,
    with $2\cdot |S_i| \le |C_i|$ and $Q_{i+1} = \ker
    \fpair{\kappa_{Q_i}, \smash{\chi_{S_i}^{C_i}}}$:
    \begin{enumerate}[topsep=1pt]
    \item For each $y\in Y$, $|\{i < k \mid y\in S_i\}| \le \log_2 |Y|$.
    \item\label{lemmaTime2} $\textsc{Split}(S_i\subseteq C_i\in
    \unnicefrac{Y}{Q_i})$ for all $0\le i < k$ takes at most
    $\CO(|E|\cdot \log |Y|)$ time in total.
    \end{enumerate}
\end{lemma}

\noindent Bringing Sections~\ref{sec:cat},~\ref{sec:opti},
and~\ref{sec:efficient} together, take a coalgebra $\xi:X\to HX$ for a
zippable \Set-functor~$H$ with a given refinement interface where
$\op{init}$ and $\op{update}$ run in linear and comparison in
constant time. Instantiate
\autoref{catPT} with the $\op{select}$ routine from
{\autoref{exampleSelects}}.\ref{exampleSelectsChi}, making
$q_{i+1} = k_{i+1} \cdot \kappa_{P_i} = \chi_{S_{i}}^{C_{i}}$ with
$2\cdot |S|\le |C|$, $S_i,C_i\subseteq X$. Replace line \ref{step:P}
by
\begin{equation}
    \unnicefrac{X}{P_{i+1}} = \textsc{Split}(\unnicefrac{X}{P_{i}},
    \unnicefrac{X}{Q_{i}}, S_{i}\subseteq C_{i}).
    \tag{\ref{kernelOptimization}'}
\end{equation}
By \autoref{thmSplitCorrect} this is equivalent to
\eqref{kernelOptimization}, and hence, by
\autoref{optimizationSummary}, to the original line \ref{step:P},
since $\chi_{S_i}^{C_i}$ respects compound blocks. By
Lemmas~\ref{lemmaInit} and~\ref{lemmaTime}.\ref{lemmaTime2}, we have

\begin{theorem} 
  The above instance of \autoref{catPT} computes the quotient modulo
  behavioural equivalence of a given coalgebra in time
  $\CO((m + n)\cdot \log n )$,
  for $m=|E|$, $n=|X|$.
\end{theorem} 
\begin{example}\label{ex:comp}
\begin{enumerate}
\item For $H=\mathcal{I}×\Potf$, we obtain the classical
  Paige-Tarjan algorithm~\cite{PaigeTarjan87} (with initial partition
  $\mathcal{I}$), with the same complexity
  $\CO((m+n)\cdot \log n)$.

\item For $HX= \mathcal{I}×\R^{(X)}$, we solve Markov chain lumping
  with an initial partition $\mathcal{I}$ in time
  $\CO((m+n)\cdot \log n)$, like the best known algorithm (Valmari and Franceschinis~\cite{ValmariF10}).

\item For infinite $A$, we need to decompose the functor $\Potf (A \times (-))$
  for labelled transition systems into $\Potf$ and $A\times(-)$, and thus obtain
  run time $\CO(m\cdot \log m)$, like in \cite{DovierEA04} but slower than
  Valmari's $\CO(m\cdot \log n)$~\cite{Valmari09}. For fixed finite $A$, we run in time
  $\CO(m\log n)$.

\item Hopcroft's classical automata minimization~\cite{Hopcroft71} is
  obtained by $HX = 2×X^A$, with running time
  $\CO(n\cdot\log n)$ for fixed alphabet $A$. For non-fixed
  $A$ the best known complexity is $\CO(|A| \cdot
  n\cdot\log n)$~\cite{Gries1973,Knuutila2001}. By using
  decomposition into $2 \times \Potf$ and $\N \times
  (-)$ we obtain $\CO(|A| \cdot n\cdot\log n + |A|\cdot
  n\cdot\log|A|)$.

\item We quotient simple (resp.~general) Segala systems~\cite{Segala95} by
  bisimilarity after decomposition into three sorts
  (cf.~\autoref{segala-twosorted}). The time bound
  $\CO((n+m)\log(n+m))$ slightly improves on the previous bound
  \cite{BaierEM00}.

\end{enumerate}
\end{example}

\section{Conclusions and Future Work}
\label{sec:conc}
We have presented a generic algorithm that quotients coalgebras by
behavioural equivalence. We have started from a category-theoretic
procedure that works for every mono-preserving functor on a category
with image factorizations, and have then developed an improved
algorithm for \emph{zippable} endofunctors on $\Set$. Provided the
given type functor can be equipped with an efficient implementation of
a \emph{refinement interface}, we have finally arrived at a concrete
procedure that runs in time $\CO((m+n)\log n)$ where $m$ is the number
of edges and $n$ the number of nodes in a graph-based representation
of the input coalgebra. We have shown that this instantiates to (minor
variants of) several known efficient partition refinement algorithms:
the classical Hopcroft algorithm~\cite{Hopcroft71} for minimization of
DFAs, the Paige-Tarjan algorithm for unlabelled transition
systems~\cite{PaigeTarjan87}, and Valmari and Franceschinis's lumping
algorithm for weighted transition systems~\cite{ValmariF10}. Moreover,
we obtain a new algorithm for simple Segala systems that is
asymptotically faster than previous
algorithms~\cite{BaierEM00}. Coverage of Segala systems is based on
modularity results in multi-sorted
coalgebra~\cite{SchroderPattinson11}.

It remains open whether our approach can be extended to, e.g., the
monotone neighbourhood functor, which is not itself zippable and also
does not have an obvious factorization into zippable functors. We do
expect that our algorithm applies beyond weighted systems.

\newpage

%
%
\bibliographystyle{plainurl}
\bibliography{refs}

\begin{thebibliography}{10}

\bibitem{Adamek05}
Ji\v{r}\'i Ad{\'{a}}mek.
\newblock Introduction to coalgebra.
\newblock {\em Theory Appl.~Categ.}, 14:157--199, 2005.

\bibitem{AdamekT90}
Ji\v{r}\'i Ad\'{a}mek and V\v{e}ra Trnkov\'a.
\newblock {\em Automata and Algebras in Categories}.
\newblock Kluwer, 1990.

\bibitem{joyofcats}
Jiří Adámek, Horst Herrlich, and George Strecker.
\newblock {\em Abstract and Concrete Categories}.
\newblock Wiley Interscience, 1990.

\bibitem{Backhouse1986}
Roland Backhouse.
\newblock {\em Program Construction and Verification}.
\newblock Prentice-Hall, 1986.

\bibitem{BaierEM00}
Christel Baier, Bettina Engelen, and Mila Majster{-}Cederbaum.
\newblock Deciding bisimilarity and similarity for probabilistic processes.
\newblock {\em J.\ Comput.\ Syst.\ Sci.}, 60:187--231, 2000.

\bibitem{BARTELS200357}
Falk Bartels, Ana Sokolova, and Erik de~Vink.
\newblock A hierarchy of probabilistic system types.
\newblock In {\em Coagebraic Methods in Computer Science, CMCS 2003}, volume~82
  of {\em ENTCS}, pages 57 -- 75. Elsevier, 2003.

\bibitem{BlomO05}
Stefan Blom and Simona Orzan.
\newblock A distributed algorithm for strong bisimulation reduction of state
  spaces.
\newblock {\em {STTT}}, 7(1):74--86, 2005.

\bibitem{Buchholz08}
Peter Buchholz.
\newblock Bisimulation relations for weighted automata.
\newblock {\em Theoret.~Comput.~Sci.}, 393:109--123, 2008.

\bibitem{clw93}
Aurelio Carboni, Steve Lack, and Robert F.~C. Walters.
\newblock Introduction to extensive and distributive categories.
\newblock {\em J. Pure Appl. Algebra}, 84:145--158, 1993.

\bibitem{CattaniS02}
Stefano Cattani and Roberto Segala.
\newblock Decision algorithms for probabilistic bisimulation.
\newblock In {\em Concurrency Theory, {CONCUR} 2002}, volume 2421 of {\em
  LNCS}, pages 371--385. Springer, 2002.

\bibitem{DerisaviEA03}
Salem Derisavi, Holger Hermanns, and William Sanders.
\newblock Optimal state-space lumping in markov chains.
\newblock {\em Inf. Process. Lett.}, 87(6):309--315, 2003.

\bibitem{DesharnaisEA02}
Josee Desharnais, Abbas Edalat, and Prakash Panangaden.
\newblock Bisimulation for labelled markov processes.
\newblock {\em Inf. Comput.}, 179(2):163--193, 2002.

\bibitem{DovierEA04}
Agostino Dovier, Carla Piazza, and Alberto Policriti.
\newblock An efficient algorithm for computing bisimulation equivalence.
\newblock {\em Theor. Comput. Sci.}, 311(1-3):221--256, 2004.

\bibitem{FislerV02}
Kathi Fisler and Moshe Vardi.
\newblock Bisimulation minimization and symbolic model checking.
\newblock {\em Formal Methods in System Design}, 21(1):39--78, 2002.

\bibitem{Gries1973}
David Gries.
\newblock Describing an algorithm by {H}opcroft.
\newblock {\em Acta Informatica}, 2:97--109, 1973.

\bibitem{gs01}
Heinz-Peter Gumm and Tobias Schr\"oder.
\newblock Monoid-labelled transition systems.
\newblock In {\em Coalgebraic Methods in Computer Science, CMCS 2001},
  volume~44 of {\em ENTCS}, pages 185--204, 2001.

\bibitem{Hopcroft71}
John Hopcroft.
\newblock An $n \log n$ algorithm for minimizing states in a finite automaton.
\newblock In {\em Theory of Machines and Computations}, pages 189--196.
  Academic Press, 1971.

\bibitem{HuynhTian92}
Dung Huynh and Lu~Tian.
\newblock On some equivalence relations for probabilistic processes.
\newblock {\em Fund.\ Inform.}, 17:211--234, 1992.

\bibitem{Jacobs17}
Bart Jacobs.
\newblock {\em Introduction to Coalgebras: Towards Mathematics of States and
  Observations}.
\newblock Cambridge University Press, 2017.

\bibitem{JacobsR97}
Bart Jacobs and Jan Rutten.
\newblock A tutorial on (co)algebras and (co)induction.
\newblock {\em Bull.~EATCS}, 62:222--259, 1997.

\bibitem{KanellakisS90}
Paris~C. Kanellakis and Scott~A. Smolka.
\newblock {CCS} expressions, finite state processes, and three problems of
  equivalence.
\newblock {\em Inf. Comput.}, 86(1):43--68, 1990.

\bibitem{KatoenEA07}
Joost{-}Pieter Katoen, Tim Kemna, Ivan Zapreev, and David Jansen.
\newblock Bisimulation minimisation mostly speeds up probabilistic model
  checking.
\newblock In {\em Tools and Algorithms for the Construction and Analysis of
  Systems, {TACAS} 2007}, volume 4424 of {\em LNCS}, pages 87--101. Springer,
  2007.

\bibitem{Klin09}
Bartek Klin.
\newblock Structural operational semantics for weighted transition systems.
\newblock In Jens Palsberg, editor, {\em Semantics and Algebraic Specification:
  Essays Dedicated to Peter D.~Mosses on the Occasion of His 60th Birthday},
  volume 5700 of {\em LNCS}, pages 121--139. Springer, 2009.

\bibitem{Knuutila2001}
Timo Knuutila.
\newblock Re-describing an algorithm by {H}opcroft.
\newblock {\em Theor.\ Comput.\ Sci.}, 250:333 -- 363, 2001.

\bibitem{KonigKupper14}
Barbara K{\"{o}}nig and Sebastian K{\"{u}}pper.
\newblock Generic partition refinement algorithms for coalgebras and an
  instantiation to weighted automata.
\newblock In {\em Theoretical Computer Science, IFIP TCS 2014}, volume 8705 of
  {\em LNCS}, pages 311--325. Springer, 2014.

\bibitem{LarsenS91}
Kim~Guldstrand Larsen and Arne Skou.
\newblock Bisimulation through probabilistic testing.
\newblock {\em Inf.\ Comput.}, 94:1--28, 1991.

\bibitem{Milner80}
Robin Milner.
\newblock {\em A Calculus of Communicating Systems}, volume~92 of {\em LNCS}.
\newblock Springer, 1980.

\bibitem{PaigeTarjan87}
Robert Paige and Robert~E. Tarjan.
\newblock Three partition refinement algorithms.
\newblock {\em SIAM J.~Comput.}, 16(6):973--989, 1987.

\bibitem{Park81}
David Park.
\newblock Concurrency and automata on infinite sequences.
\newblock In {\em Theoretical Computer Science, 5th GI-Conference}, volume 104
  of {\em LNCS}, pages 167--183. Springer, 1981.

\bibitem{RanzatoT08}
Francesco Ranzato and Francesco Tapparo.
\newblock Generalizing the {P}aige-{T}arjan algorithm by abstract
  interpretation.
\newblock {\em Inf.\ Comput.}, 206:620--651, 2008.

\bibitem{Rutten00}
Jan Rutten.
\newblock Universal coalgebra: a theory of systems.
\newblock {\em Theor.\ Comput.\ Sci.}, 249:3--80, 2000.

\bibitem{SchroderPattinson11}
Lutz Schr{\"o}der and Dirk Pattinson.
\newblock Modular algorithms for heterogeneous modal logics via multi-sorted
  coalgebra.
\newblock {\em Math.\ Struct.\ Comput.\ Sci.}, 21(2):235--266, 2011.

\bibitem{Segala95}
Roberto Segala.
\newblock {\em Modelling and Verification of Randomized Distributed Real-Time
  Systems}.
\newblock PhD thesis, Massachusetts Institute of Technology, 1995.

\bibitem{Valmari09}
Antti Valmari.
\newblock Bisimilarity minimization in {$\CO(m \log n)$} time.
\newblock In {\em Applications and Theory of Petri Nets, {PETRI} {NETS} 2009},
  volume 5606 of {\em LNCS}, pages 123--142. Springer, 2009.

\bibitem{ValmariF10}
Antti Valmari and Giuliana Franceschinis.
\newblock Simple {$\CO(m\log n)$} time {M}arkov chain lumping.
\newblock In {\em Tools and Algorithms for the Construction and Analysis of
  Systems, TACAS 2010}, volume 6015 of {\em LNCS}, pages 38--52. Springer,
  2010.

\bibitem{Benthem77}
Johann van Benthem.
\newblock {\em Modal Correspondence Theory}.
\newblock PhD thesis, Universiteit van Amsterdam, 1977.

\bibitem{MeydenZ07}
Ron van~der Meyden and Chenyi Zhang.
\newblock Algorithmic verification of noninterference properties.
\newblock In {\em Views on Designing Complex Architectures, VODCA 2006}, volume
  168 of {\em ENTCS}, pages 61--75. Elsevier, 2007.

\bibitem{Worrell05}
James Worrell.
\newblock On the final sequence of a finitary set functor.
\newblock {\em Theor.\ Comput.\ Sci.}, 338:184--199, 2005.

\bibitem{ZhangEA08}
Lijun Zhang, Holger Hermanns, Friedrich Eisenbrand, and David Jansen.
\newblock {F}low {F}aster: Efficient decision algorithms for probabilistic
  simulations.
\newblock {\em Log.\ Meth.\ Comput.\ Sci.}, 4(4), 2008.

\end{thebibliography}

\newpage
\renewcommand{\thesection}{\Alph{section}}
\makeatletter
\setcounter {section}{0}\setcounter {subsection}{0}\gdef \thesection {\@Alph \c@section }
\makeatother
\section{Omitted Details and Proofs}

\subsection*{Details for \autoref{sec:prelim}}

\begin{remark}\label{pullbackOfKernelPair}
  \begin{enumerate}[label=(\arabic*)]
%
  \item The following \emph{diagonalization property} holds for image
    factorizations: given a commutative square $m\cdot f= g\cdot e$
    where $m$ is a monomorphism and $e$ a regular epimorphism, there
    exists a (necessarily unique) \emph{diagonal} $d$ such that
    $m \cdot d = g$ and $d \cdot e = f$. In particular, image
    factorizations are unique up to isomorphism, so $\Im(f)$ is
    well-defined up to isomorphism.

  \item For any object $X$ there is a bijective correspondence between
    kernels of morphisms $f$ with domain $X$ and regular quotients of
    $X$. Indeed, in one direction take the coequalizer of a given
    kernel pair $\ker f$. In the reverse direction, take the kernel of
    a given regular epimorphism $e: X \epito Y$
    (see~\cite[Prop.~11.22(2)]{joyofcats}).
    
  \item It follows that two morphisms $f: X \to Y$ and $g: X \to Y'$
    have the same kernel iff they have the same image:
    \[
      \Im f = \Im g \qquad\iff \qquad \ker f = \ker g.
    \]
    To see this, take the image factorizations of $f = m \cdot e$ and
    $g = m'\cdot e'$, respectively, and use that $\ker f = \ker e$ and
    $\ker g = \ker e'$. 


  \item A \emph{relation} is a jointly monic parallel pair of
    morphisms $f,g: E\rightrightarrows X$ (not necessarily a kernel
    pair). We write $\kappa_E: X \twoheadrightarrow \unnicefrac{X}{E}$
    for their coequalizer; we refer to the object $\unnicefrac{X}{E}$ as
    the \emph{quotient} of $X$ modulo $E$, and to $\kappa_E$ as the
    \emph{quotient map}. Indeed, in $\Set$, $\unnicefrac{X}{E}$ is the usual
    quotient of $X$ modulo the equivalence relation generated by
    $\{(fx,gx) \mid x \in E\}$. When $f$ and $g$ and $X$ are clear
    from the context we just write the object $E$ for the relation.

  \item We say that a morphism $f: X\to Y$ is \emph{well-defined} on
    (the equivalence classes of) a relation
    $\pi_1, \pi_2: E \rightrightarrows X$ if
    \[
    \begin{mytikzcd}
        E \arrow{r}{\pi_1}
        \arrow{d}{\pi_2}
        & X \arrow{d}{f}
        \\
        X \arrow{r}{f}
        & Y
      \end{mytikzcd}
    \]
    commutes. Then by the universal property of
    $\kappa_E: X \to \unnicefrac{X}{E}$ we obtain a unique morphism
    $f': \unnicefrac{X}{E}\to Y$ such that $f = f' \cdot \kappa_E$. In
    $\Set$, this is the usual well-definedness of the map~$f$ on the
    equivalence classes in $\unnicefrac{X}{E}$ witnessed by the map
    $f'$.

  \item Following the standard terminology in $\Set$, we say that a
    quotient $\unnicefrac{X}{E_1}$ is \emph{finer} than (or a
    \emph{refinement} of) another quotient $\unnicefrac{X}{E_2}$ if the
    quotient map $\kappa_{E_2}$ is well-defined on~$E_1$. This induces
    a refinement relation on kernels, described as follows:
    $\pi_1, \pi_2: E \rightrightarrows X$ is \emph{finer} than
    $\pi_1', \pi_2': E' \rightrightarrows X$ if there exists a
    morphism $m: E \to E'$ (necessarily monic) such that
    $\pi_i' \cdot m = \pi_i$, $i = 1,2$.
  \end{enumerate}
\end{remark}

\begin{lemma}\label{kernelMono} We have $\ker(m\cdot f) = \ker f$ for
  every $f: X \to Y$ and every monic $m: Y \monoto Z$.
\end{lemma}
\begin{proof}
    This can be shown by checking directly that $\ker f$ is the kernel
    of $m\cdot f$ and using that $m$ is monic.
\end{proof}

\begin{lemma} \label{lem:regepi-pushout}
  $\C$ has pushouts of regular epimorphisms (i.e.\ of spans containing
  at least one regular epimorphism).
\end{lemma}
\begin{proof}
  Let $X\xleftarrow{e}Y\xrightarrow{h}W$ be a span, with $e$ a
  regular epi. Let $(\pi_1,\pi_2)$ be the kernel pair of $e$, and
  let $q:W\to Z$ be the coequalizer of $h\pi_1$ and $h\pi_2$
  (both exist by our running assumptions). Then $e$ is the
  coequalizer of $\pi_1,\pi_2$, so that there exists $r:X\to Z$ such
  that $re=qh$. We claim that
  \begin{equation*}
    \begin{mytikzcd}
      Y \arrow[->>]{r}[above]{e} \arrow{d}[left]{h} & 
      X\arrow{d}[right]{r} \\
      W\arrow[->>]{r}[below]{q} & Z
    \end{mytikzcd}
  \end{equation*}
  is a pushout. Uniqueness of mediating morphisms is clear since $q$
  is epic; we show existence. So let
  $W\xrightarrow{g} U \xleftarrow{f} X$ be a competitor, i.e.\
  $fe=gh$. Then $gh\pi_1=fe\pi_1=fe\pi_2=gh\pi_2$, so by
  the coequalizer property of $q$ we obtain $k:Z\to U$ such that
  $kq=g$. It remains to check that $kr=f$. Now
  $kre=kqh=gh=fe$, which implies the claim because $e$ is
  epic.
\end{proof}

\begin{lemma}\label{lem:coalg-colims}
  $\Coalg(H)$ has all coequalizers and pushouts of regular
  epimorphisms.
\end{lemma}
\begin{proof}
Since the forgetfulfunctor $\Coalg(H) \to \C$ creates all colimits, the
statement follows directly by our running assumptions and
Lemma~\ref{lem:regepi-pushout}.
\end{proof}

\paragraph*{Proof of \autoref{lemmaSimple}}
Uniqueness up to isomorphism means:
\begin{lemma}
  Let $(C,c)$ be a coalgebra, and let $e_i:(C,c)\to(D_i,d_i)$,
  $i=1,2$, be quotients with $(D_i,d_i)$ simple. Then $(D_1,d_1)$ and
  $(D_2,d_2)$ (more precisely the quotients $e_1$ and $e_2$) are
  isomorphic.
\end{lemma}
\begin{proof}
  By Lemma~\ref{lem:coalg-colims}, there is a pushout
  $D_1\xrightarrow{f_1}E\xleftarrow{f_2}D_2$ of
  $D_1\xleftarrow{e_1}C\xrightarrow{e_2}D_2$ in $\Coalg(H)$. Since
  regular epimorphisms are generally stable under pushouts,
  $f_1$ and $f_2$ are regular
  epimorphisms, hence isomorphisms because $D_1$ and $D_2$ are simple;
  this proves the claim.
\end{proof}

\paragraph*{Behavioural equivalence between coalgebras}
\begin{remark}
  Using elementwise notation for intuition, `elements' $x\in C$ and
  $y\in D$ of coalgebras $(C,c)$ and $(D, d)$ are \emph{behaviourally
    equivalent} (written $x \sim y$) if they can be merged by coalgebra
  morphisms: $x \sim y$ iff there exists a coalgebra $(E,e)$ and
  coalgebra morphisms $f:(C,c)\to (E,e)$, $g: (D,d) \to (E,e)$ such that $f(x)=g(y)$. Under
  our running assumptions, any two behaviourally equivalent elements
  can be identified under a regular quotient, so that a simple
  quotient of a coalgebra already identifies all behaviourally
  equivalent elements: Reformulated in proper categorical terms, we
  claim that every pullback of two coalgebra morphisms
  $f,g:(C,d)\to (D,d)$ is contained in the kernel pair of some
  morphism $e:(C,d)\to (E,e)$.  Indeed, by
  Lemma~\ref{lem:coalg-colims} we can take $e=qf=qg$ where $q$ is the
  coequalizer of $f$ and $g$ in $\Coalg(H)$.
  
  A \emph{final coalgebra} is a terminal object in the category of
  coalgebras, i.e.\ a coalgebra $(C,c)$ such that every coalgebra
  $(D,d)$ has a unique coalgebra morphism into $(C,c)$.  There are
  reasonable conditions under which a final coalgebra is guaranteed to
  exist, e.g.\ when $\C$ is a locally presentable category (in
  particular, when $\C=\Set$) and $H$ is accessible. If $(C,c)$ is a
  final coalgebra and $H$ preserves monos, then we can describe the
  simple quotient of a coalgebra $(D,d)$ as the image of $(D,d)$ under
  the unique morphism into $(C,c)$; in particular, in this case every
  coalgebra has a simple quotient.
\end{remark}

\subsection*{Details for \autoref{sec:cat}}

\paragraph*{Notes on \autoref{ass:sec3}}
For $\C = \Set$, the assumption that $H$ preserves monos is
w.l.o.g. First note that every endofunctor on sets preserves non-empty
monos. Moreover, for any set functor $H$ there exists a set functor
$H'$ that is naturally isomorphic to $H$ on the full subcategory of
all non-empty sets~\cite[Theorem~3.4.5]{AdamekT90}, and hence has
essentially the same coalgebras as $H$ since there is only one
coalgebra structure on $\emptyset$.

\paragraph*{Proof of Lemma~\ref{lem:inc}}
\begin{enumerate}
\item $P_{i+1}$ finer than $P_i$ and $Q_{i+1}$ finer than $Q_i$:
  Let $p: \prod_{j\le i+1} K_j \to \prod_{j\le i} K_j$ be the product
  projection. Clearly we have $\bar q_i = p \cdot \bar q_{i+1}$ and
  therefore, for the kernel pair
  $\pi_1,\pi_2: Q_{i+1} \rightrightarrows X$ we clearly have
  \[
  \bar q_i \cdot \pi_1 = p \cdot \bar q_{i+1} \cdot \pi_1 = p \cdot
  \bar q_{i+1} \cdot \pi_2 = \bar q_i \cdot \pi_2.
  \]
  Hence, we obtain a unique $Q_{i+1} \to Q_i$ commuting with the
  projections of the kernel pairs.

  Similarly, for the kernel pair
  $\pi_1, \pi_2: P_{i+1} \rightrightarrows X$ we have
  \[
  (H\bar q_i \cdot \xi) \cdot \pi_1 = H p \cdot H\bar q_{i+1} \cdot
  \xi \cdot \pi_1 = H p \cdot H\bar q_{i+1} \cdot \xi \cdot \pi_2 =
  (H\bar q_i \cdot \xi) \cdot \pi_2.
  \]
  Thus, there exists a unique morphism $P_{i+1} \to P_i$ commuting
  with the kernel pair projections.

\item $P_i$ finer than $Q_{i+1}$:
  Induction on $i$. Since $Q_{i+1}=\ker\bar q_{i+1}$ and
  $\bar q_{i+1}=\fpair{q_0,\dots,q_{i+1}}$, it suffices to show that
  $P_i$ is finer than $\ker q_j$ for $j=0,\dots,i+1$. For $j\le i$, we
  have by Lemma~\ref{lem:inc} that $P_i$ is finer than $P_j$, which is
  finer than $\ker q_j$ by induction. Moreover, $P_i$ is finer than
  $\ker q_{i+1}$ because $q_{i+1}$ factors through
  $X\to\unnicefrac{X}{P_i}$ by construction.
\end{enumerate}

\paragraph*{Proof of \autoref{propQuot}}
  Since $Q_i=\ker\bar q_i$, the image factorization of $\bar q_i$ has
  the form $\bar q_i=m \cdot \kappa_{Q_i}$. By definition of $P_i$ and
  since $H$ preserves monos, we thus have
  $P_i = \ker(H\bar q_i\cdot \xi)=\ker(H\kappa_{Q_i}\cdot\xi)$, and
  hence obtain $\unnicefrac{\xi}{Q_i}$ as in~\eqref{eq:xiQuotient} by
  the coequalizer property of $\kappa_{P_i}$. 

\paragraph*{Proof of \autoref{soundness}}
  We claim that
  \begin{equation}\label{eq:hq-to-hp}
    \text{if $\ker h$ is finer than $Q_i$, then $\ker h$ is finer than $P_i$}.
  \end{equation}
  This is seen as follows: If $\ker h$ is finer than $Q_i$, then
  $\kappa_{Q_i}:X\to\unnicefrac{X}{Q_i}$ factors through $h:X\to D$, so
  that $H\kappa_{Q_i}\cdot\xi$ factors through $Hh\cdot\xi$ and hence
  through $h$, since $Hh\cdot\xi=d\cdot h$. Since
  $P_i=\ker(H\kappa_{Q_i}\cdot\xi)$, this implies that $\ker h$ is
  finer than $P_i$.  The claim of the lemma is then proved by
  induction: for $i=0$, the claim for $Q_0=X\times X$ is trivial, and
  the one for $P_0$ follows by~\eqref{eq:hq-to-hp}. The inductive step
  is by \autoref{PfinerthanQ} and~\eqref{eq:hq-to-hp}.

\paragraph*{Proof of \autoref{correctness}}
  Let $h: (\unnicefrac{X}{Q_i},\unnicefrac{\xi}{Q_i}) \to (D,d)$ be a
  quotient. Then $h\cdot \kappa_{Q_i}: (X,\xi) \to (D,d)$ is a
  quotient of $(X,\xi)$, so by \autoref{soundness}, $\ker( h\cdot\kappa_{Q_i})$
  is finer than $Q_i$. Of course, $Q_i$ is also
  finer than $\ker( h\cdot\kappa_{Q_i})$, so~$h$ is an isomorphism.

\paragraph*{Proof of \autoref{finalchain}}
    If we have
    $k_{i+1} := \id_{\unnicefrac{X}{P_i}}$, then the inclusions
    $P_i\monoto Q_{i+1}$
    become isomorphisms: we have
    $q_{i+1} = \kappa_{P_i}:X\to \unnicefrac{X}{P_i}$ for all $i$, so
    the $q_i$ successively refine each other, so that
    $Q_{i+1}=\ker\bar q_{i+1}=\ker\fpair{q_0,\dots,q_{i+1}}=\ker
    q_{i+1}=P_i$.

    We show $Q_i = \ker \xi^{(i)}$ for $i\ge 0$
    by induction on $i$, with trivial base case. For the
    inductive step, first note that from $\ker\bar q_{i}=\ker q_{i}$
    and because $q_{i}$ is a regular epi, we obtain a mono $m$ such
    that $\bar q_{i+1}=m q_{i+1}$; similarly, the inductive hypothesis
    implies that we have a mono $n$ such that $\xi^{(i)}=n q_i$. Since
    $H$ preserves monomorphisms, this implies that
    \begin{equation*}
      Q_{i+1}=P_i=\ker(H\bar q_i\xi)=\ker(H
      q_i\xi)=\ker(H\xi^{(i)}\xi)=\ker(\xi^{(i+1)}).
    \end{equation*}

\paragraph*{Proof of \autoref{noprogress}}

First note that \op{select} does not retain any new information in
$k_{i+1}$ iff $q_{i+1}=k_{i+1}\kappa_{P_i}$ factors through
$f_i\kappa_{P_i}=\kappa_{Q_i}$. Now we reason as follows:
$Q_i=\ker\bar q_i$ is finer than
$Q_{i+1}=\ker\fpair{\bar q_i,q_{i+1}}$ iff $Q_i$ is finer than
$\ker q_{i+1}$ iff $q_{i+1}$ factors through $\kappa_{Q_i}$.

\subsection*{Details for \autoref{sec:opti}}

\paragraph*{Proof of \autoref{lem:closure}}
  Let $F, G$ be endofunctors.

(1)~Suppose that both $F$ and $G$ are zippable. To see that $F \times
G$ is zippable one uses that monos are closed under products:
\[
\begin{mytikzcd}[column sep=15mm,row sep=5mm]
    F(A+B) × G(A+B)
    \ar[r,>->,"\op{unzip}_{F,A,B} × \op{unzip}_{G,A,B}"{yshift=2mm}]
    \ar[to path={
            |- (\tikztotarget) \tikztonodes
        },
        rounded corners,
        ]{dr}[below,near end]{\op{unzip}_{F×G,A,B}}
    &
    F(A+1) × F(1+B)
    ×G(A+1) × G(1+B)
    \ar[phantom,d,"\cong" {sloped}]
    \\
    &
    \big(F(A+1)
    ×G(A+1)\big)
    × \big(F(1+B)
    × G(1+B)\big)
\end{mytikzcd}
\]

(2)~Suppose again that $F$ and $G$ are zippable. To see that $F+G$ is zippable consider the diagram below: 
\[
\begin{mytikzcd}[column sep=15mm,row sep=5mm]
    F(A+B) + G(A+B)
    \ar[r,>->,"\op{unzip}_{F,A,B} + \op{unzip}_{G,A,B}"{yshift=2mm}]
    \ar[to path={
            |- (\tikztotarget) \tikztonodes
        },
        rounded corners,
        ]{dr}[below,near end]{\op{unzip}_{F×G,A,B}}
    &
    \big(F(A+1) × F(1+B)\big)
    + \big(G(A+1) × G(1+B)\big)
    \ar[>->,d,"\fpair{(\pi_1+\pi_1), (\pi_2+\pi_2)}" {right}]
    \\
    &
    \big(F(A+1)
    +G(A+1)\big)
    × \big(F(1+B)
    + G(1+B)\big)
\end{mytikzcd}
\]
The horizontal morphism is monic since monos are closed under
coproducts in $\C$. 
The vertical morphism is monic since for any sets
$A_i$ and $B_i$, $i = 1, 2$, the following morphism clearly is a
monomorphism:
\[
  (A_1×B_1) + (A_2×B_2)
  \xrightarrow{\fpair{(\pi_1+\pi_1), (\pi_2+\pi_2)}}
  (A_1+A_2) × (B_1+B_2)
\]

(3)~Suppose now that $F$ is a subfunctor of $G$ via $s: F
\rightarrowtail G$, where $G$ is zippable. Then the following diagram
shows that $F$ is zippable, too:
\[
    \begin{mytikzcd}[column sep = 2cm]
        F(A+B)
            \arrow{r}{\op{unzip}_{F,A,B}}
            \arrow[>->]{d}[left]{s_{A×B}}
            \arrow[>->,
                  to path={
                    -- ([xshift=-4mm]\tikztostart.west)
                    |- ([yshift=-4mm]\tikztotarget.south)
                    -- (\tikztotarget)
                  },
                  rounded corners,
                  ]{dr}{}
        & F(A+1) × F(1+B)
            \arrow{d}{s_{A+1} × s_{1+B}}
        \\
        G(A+B)
            \arrow[>->]{r}{\op{unzip}_{G,A,B}}
        & G(A+1) × G(1+B)
    \end{mytikzcd}
\]

\paragraph*{Proof of \autoref{lem:additive}}
  Let $\alpha:H(X+Y) \rightarrowtail HX × HY$ be componentwise
  monic. Then the square
    \[
    \begin{mytikzcd}[column sep = 3cm]
        H(A+B) \ar[d,>->,"\alpha_{A,B}"'] \ar[r,"\op{unzip}"]
        & H(A+1)×H(1+B)
          \ar[d,"\alpha_{A,1}×\alpha_{1,B}"]
        \\
        HA×HB
          \ar[r,>->,"\fpair{HA×H!,H!×HB}"]
        & (HA×H1)×(H1×HB)
    \end{mytikzcd}
    \]
    commutes, by naturality of $\alpha$ in each of the components. The
    bottom morphism is monic because it has a left inverse,
    $\pi_1×\pi_2$. Therefore, $\op{unzip}$ is monic as well.

\begin{remark}\label{rem:extensive}
  Out of the above results, only zippability of the identity and
  coproducts of zippable functors depend on our assumptions on $\C$
  (see~\autoref{ass:C}). Indeed, zippable functors are closed under
  coproducts as soon as monomorphism are closed under coproducts,
  which is satisfied in most categories of interest. Zippability of
  the identity holds whenever $\C$ is extensive, i.e.~it has
  well-behaved set-like coproducts. Formally, a category is
  extensive~\cite{clw93} if it has finite coproducts and pullbacks
  along coproduct injections such that coproducts are
  \begin{enumerate}[label=(\arabic*)]
  \item \emph{disjoint}, i.e., coproduct injections are monomorphic
    and the pullback of distinct coproduct injections is $0$ (the
    initial object),
  \item \emph{universal}, i.e., the pullbacks of a morphism
    $h: Z \to A+B$ along the coproduct injections, yields a
    coproduct $Z = X + Y$ and $h = f + g$:
    \[
      \begin{mytikzcd}
        X \ar[r, "x"]
        \ar[d, "f"]
        & Z
        \ar[d, "h"]
        &\ar[l, "y"]
        Y 
        \ar[d, "g"]
        \\
        A
        \ar[r, "\inl"]
        & 
        A + B
        &
        B
        \ar[l, "\inr"]
      \end{mytikzcd}
    \]
  \end{enumerate}
  In an extensive category coproducts commute with pullbacks, and
  therefore monomorphisms are closed under coproducts. 

  Examples of extensive categories are the categories of sets, posets and
  graphs as well as any presheaf category. In addition, the categories
  of unary algebras and of J\'onsson-Tarski algebras (i.e.~algebras
  $A$ with one binary operation $A \times A \to A$ that is an
  isomorphism) are extensive. More generally, any topos is
  extensive.

  The category of monoids is not extensive.
\end{remark}


\paragraph*{Details for \autoref{non-zippable}}
The following example shows that the optimized algorithm is not correct for the
non-zippable functor $\Potf\Potf$. The $\op{select}$ routine here $\chi_S^C$ even
fulfills the latter assumption that $\op{select}$ respects compounds blocks
(\autoref{defRespectCompounds}).
\begin{example}
Consider the following coalgebra $\xi: X\to HX$ for $HX = 2×\Potf\Potf X$:
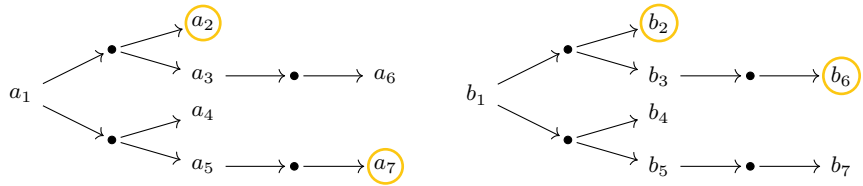
\begin{center}
\begin{tikzpicture}[
        n/.style = {
               execute at begin node=\(,%
               execute at end node=\),%
               inner sep = 1mm,
        },
    ]
    \begin{scope}[
        grow'                   = right,
        sibling distance        = 7mm,
        level distance          = 12mm,
        edge from parent/.style = {
            draw,
            ->,
            shorten <= 1pt,
            shorten >= 1pt,
        },
        level 1/.append style = {
            sibling distance = 12mm,
        },
        level 2/.append style = {
            sibling distance = 7mm,
        },
        every node/.style       = {n,font=\footnotesize},
        sloped,
        dummy/.style    = {circle,
                           draw=black,
                           fill=black,
                           inner sep=1pt,
                           outer sep=1pt,
                           minimum width=0,
                           minimum height=0,
        },
        final/.style    = {
            circle,
            draw=lipicsYellow,
            line width=1pt,
            inner sep =1pt,
        },
    ]
        \node (a1) {a_1}
        child { node[dummy] {}
            child { node[final] (a2) {a_2} }
            child { node (a3) {a_3}
                child { node[dummy] {}
                    child { node (a6) {a_6} }
                }
            }
        }
        child { node[dummy] {}
            child { node (a4) {a_4} }
            child { node (a5) {a_5}
                child { node[dummy] {}
                    child { node[final] (a7) {a_7} }
                }
            }
        }
        ;
        \node (b1) at ([xshift=6cm]a1){b_1}
        child { node[dummy] {}
            child { node[final] (b2) {b_2} }
            child { node (b3) {b_3}
                child { node[dummy] {}
                    child { node[final] (b6) {b_6} }
                }
            }
        }
        child { node[dummy] {}
            child { node (b4) {b_4} }
            child { node (b5) {b_5}
                child { node[dummy] {}
                    child { node (b7) {b_7} }
                }
            }
        }
        ;
    \end{scope}
\end{tikzpicture}
\end{center}
States $x$ with $\pi_1(\xi(x)) = 1$ are indicated by the circle. When computing
only
\begin{equation}
  P_{i+1}' := \ker \fpair{H\bar q_i\cdot \xi,Hq_{i+1}\cdot\xi}
  = \ker P_i' \cap \ker (Hq_{i+1}\cdot\xi)
  \tag*{\eqref{kernelOptimization}}
\end{equation}
instead of $P_i$, then $a_1$ and $b_1$ are not distinguished,
although they are behaviourally different.
\end{example}
In order to simplify the partitions, we define abbreviations for the circle and
non-circle states without successors, and the rest:
\[
    F:= \{a_2,a_7, b_2,b_6\}
    \quad
    N:=\{a_4,a_6, b_4,b_7\}
    \quad
    C:=\{a_1,a_3,a_5, b_1,b_3,b_5\}
\]
Running the optimized algorithm, i.e.~computing $Q_i$ and $P_{i}'$, one obtains
the following sequence of partitions.
\renewcommand{\arraystretch}{1.2}%
\newcommand{\myblocklist}[1]{
    \big\{ #1
    \big\}
  }%
\newcommand{\mytopbotrule}{
    \noalign{
\global\dimen1\arrayrulewidth
\global\arrayrulewidth1.3pt
}\hline
\noalign{
\global\arrayrulewidth\dimen1 
}
  }%
\begin{center}%
\begin{tabular}{@{\hspace{0mm}}LL@{\hspace{1mm}}L@{\hspace{1mm}}L@{\hspace{0mm}}}
  \mytopbotrule
  i & q_i & \unnicefrac{X}{Q_i} & \unnicefrac{X}{P_i'}
\\ \hline
\\[-4.5mm]
    0
    & !: X\to 1
    & \myblocklist{X}
    &
        \myblocklist{F, N,
        C }
    \\[1mm]
    1
    & \kappa_{P_0'}: X\twoheadrightarrow \unnicefrac{X}{P_0'}
    & \myblocklist{F,N, C}
    & \myblocklist{F,N, \{a_1,b_1\}, \{a_3,b_5\}, \{a_5,b_3\}}
\\[1mm]
    2
    & \chi_S^C: X\to 3,\text{ for } S = \{a_3,b_5\}
    & \myblocklist{F,N, \{a_1,b_1\}, \{a_3,b_5\}, \{a_5,b_3\}}
    & \myblocklist{F,N, \{a_1,b_1\}, \{a_3,b_5\}, \{a_5,b_3\}}
\\[1mm] \mytopbotrule
\end{tabular}
\end{center}
Note that in the step $i=2$ one obtains the same result for $S' :=\{a_5,b_3\}$
or $S'' := \{a_1,b_1\}$. For $S$ as in the table, $a_1$ and $b_1$ are not split
in $\unnicefrac{X}{P_2'}$ because:
\begin{align*}
  H\chi_S^C \cdot \xi(a_1)
  &= H\chi_S^C \myblocklist{\{a_2,a_3\}, \{a_4,a_5\}} \\
    &= \phantom{H\chi_S^C}\myblocklist{\{0,2\}, \{0,1\}}
    \\ &
    = \phantom{H\chi_S^C}\myblocklist{\{0,1\}, \{0,2\}} \\
    &= H\chi_S^C \myblocklist{\{b_2,b_3\}, \{b_4,b_5\}}
    = H\chi_S^C \cdot \xi(b_1)
\end{align*}
Now the algorithm terminates because $\unnicefrac{X}{Q_2} =
\unnicefrac{X}{P_2}$, but without distinguishing $a_1$ from $b_1$.

\paragraph*{Proof of \autoref{compoundBlockEquivalences}}

\begin{proof}
    4.~$\Rightarrow$ 1. In \Set, kernels are equivalence relations. Obviously, $\ker
    a\cup \ker b$ is both reflexive and symmetric. For transitivity,
    take $(x,y),(y,z) \in \ker a\cup \ker b$. Then $x, z \in [y]_a\cup
    [y]_b$. If $[y]_a\subseteq [y]_b$, then $x,z \in [y]_b$ and $(x,z)
    \in \ker b$; otherwise $(x,z) \in \ker a$.
    
    1.~$\Rightarrow$ 2. In \Set, monomorphisms are stable under pushouts, so it is
    sufficient to show that $\ker a \cup \ker b$ is the kernel of the
    pushout of the epi-parts of $a$ and $b$. In other words, w.l.o.g.~we may 
    assume that $a$ and $b$ are epic, and we need to check that $\ker
    a\cup \ker b$ is the kernel of $p:= p_A\cdot a = p_B\cdot b$, with
    \[
    \begin{mytikzcd}[baseline=(bot.base), ampersand replacement=\&]
        D \ar[->>,d,swap,"b"] \ar[->>,r,"a"]
        \& A
        \ar[d,"p_A"]
        \\
        |[alias=bot]|
        B
        \ar[r,"p_B"]
        \& P.
        \pullbackangle{135}
    \end{mytikzcd}
    \]
    Let $\ker a\cup \ker b$ be the kernel of some $y: D\to Y$. Then,
    $y$ makes the projections of $\ker a$ (resp.~$\ker b$) equal and hence
    the coequalizer $a$ (resp.~$b$) induces a unique $y_A$ (resp.~$y_B$):
    \[
    \begin{mytikzcd}
        \ker a
        \ar[r,hook]
        \ar[rr,shiftarr={yshift=5mm},"\pi_1"]
        \ar[rr,shiftarr={yshift=-5mm},"\pi_2"']
        &
        \ker a \cup \ker b
        \ar[r,shift left = 1,"{\pi_1}"]
        \ar[r,shift left = -1,"{\pi_2}"']
        & D
        \ar[r,"y"]
        \ar[dr,->>,"a"']
        & Y
        \\
        &&& A
        \ar[u,dashed,"\exists! y_A"']
    \end{mytikzcd}
    \quad
    \begin{mytikzcd}
        \ker b
        \ar[r,shift left = 1,"\pi_1"]
        \ar[r,shift left = -1,"\pi_2"']
        & D
        \ar[r,"y"]
        \ar[dr,->>,"b"']
        & Y
        \\
        && B
        \ar[u,dashed,"\exists! y_B"']
    \end{mytikzcd}
    \]
    Because of $y_B \cdot b = y = y_A \cdot a$,  $(y_A,y_B)$ is a competing
    cocone for the pushout. This induces a cocone morphism $y_P: (P,p_A,p_B)\to
    (Y,y_A,y_B)$, and we have 
    \[
        y_P\cdot p = y_P\cdot p_A\cdot a = y_A\cdot a = y.
    \]
    With this, we are ready to show that $\ker a\cup \ker b$ is a kernel for
    $p$. Consider two morphisms $c_1: C\to D$, $c_2: C\to D$ with $p\cdot c_1 =
    p\cdot c_2$, then we have
    \[
        y\cdot c_1 = y_P\cdot p \cdot c_1 = y_P\cdot p \cdot c_2 =  y\cdot c_2.
    \]
    This induces a unique cone morphism $C\to \ker a \cup \ker b$ as desired.

    2. $\Rightarrow$ 3. Take $x,y,z \in D$ with $a(x)=a(y)$ and
    $b(y) = b(z)$. Then $a(x)$ and $b(z)$ are identified in the
    pushout $P$: 
    \[
      p(x) = p_A \cdot a(x) = p_A \cdot a(y) = p_B\cdot b(y) = p_B
      \cdot b(z) = p(z). 
    \]
    This shows that $(x,z)$ lies in $\ker a\cup \ker b$, hence we have
    that $a(x) = a(z)$ or $b(x) = b(z)$.

    3. $\Rightarrow$ 4. For a given $y\in D$, there is nothing to show
    in the case where $[y]_a \subseteq [y]_b$. Otherwise if
    $[y]_a \not\subseteq [y]_b$, then there is some $x\in [y]_a$,
    i.e.~such that $a(x) = a(y)$, with $b(x) \neq b(y)$. Now let
    $z \in [y]_b$, i.e.~$b(y) = b(z)$. Then, by assumption,
    $a(x) = a(y) = a(z)$ or $b(x) = b(y) = b(z)$. Since the latter
    does not hold, we have $a(y) = a(z)$, i.e.~$z\in [y]_a$.
\end{proof}

\paragraph*{Proof of \autoref{examplesRespectCompoundBlocks}}
\begin{enumerate}
\item For $S\in Y$ and $C := [S]_z$, $\chi_{\{S\}}^C$ respects
  compound blocks by
  Lemma~\ref{compoundBlockEquivalences}.\ref{classInclusion}:
    \begin{itemize}
        \item For $p \in Y\setminus C$, $z(p) \neq z(S)$ and so $[p]_z
        \subseteq Y\setminus C = [p]_k$.

        \item For $p \in C$, $z(p) = z(S)$ and so $[p]_k \subseteq C =
        [p]_z$.
    \end{itemize}

\item The \op{select} routine returning the identity respects compound
blocks, because for any morphism $a: D\to A$, $\ker a\cup \ker \id_D =
\ker a$ is a kernel.

\item The constant $k=\mathbin{!}$ respects compound blocks, because for all
$p\in Y$: $[p]_z \subseteq Y = [p]_!$.
\end{enumerate}

\paragraph*{Proof of \autoref{propZippable}}

\begin{lemma} \label{zippableExtended}
    Let $H$ be zippable and $f: A\to C$, $g: B\to D$. Then 
    \[
        H(A+B) \xrightarrow{\fpair{H(A+g),H(f+B)}} H(A+D) × H (C+B)
    \]
    is monic.
\end{lemma}
\begin{proof}
    By finality of $1$, the diagram
    \[
    \begin{mytikzcd}[column sep = 3cm]
        H(A+B)
        \arrow{d}[left]{\fpair{H(A+g),H(f+B)}}
        \arrow[>->]{dr}[sloped,above]{
            \op{unzip}_{H,A,B}
            = \fpair{H(A+!),H(!+B)}
        }
        \\
        H(A+D) × H (C+B)
        \arrow{r}{H(A+!) × H (!+B)}
        &
        H(A+1) × H (1+B)
    \end{mytikzcd}
    \]
    commutes. Since the diagonal arrow is monic, so is
    $\fpair{H(A+g),H(f+B)}$.
\end{proof}

\begin{proof}[Proof of \autoref{propZippable}]
    Define
    \begin{align*}
        D_A &= \{ x\in D\mid [x]_a \subseteq [x]_b \}
        \overset{d_A} \hookrightarrow D
    &
        D_B &= \{ x\in D\mid [x]_b \subsetneq [x]_a \}
        \overset{d_B} \hookrightarrow D
    \end{align*}
    By construction and Lemma~\ref{compoundBlockEquivalences}, we have
    the isomorphism
    $\phi = [d_A, d_B]^{-1}: D \overset{\cong}{\to} D_A+D_B$. We
    denote the images of the restrictions of $a$ and $b$ to $D_A$ and
    $D_B$, respectively, by $a': D_A \twoheadrightarrow A'$,
    $b': D_B \twoheadrightarrow B'$. We claim that we can
    define maps $c_A$ and $c_B$ such that
    \[
    \begin{mytikzcd}
    D \arrow{d}[left]{a}
    & D_A
        \arrow[hook]{rr}{d_A}
        \ar[d,->>,swap,"a'"]
        \arrow[hook']{l}[above]{d_A}
    && D
        \arrow{d}{b}
    \\
    A & \arrow[hook']{l}[above]{i_A}
    A'
        \arrow{r}{c_B}
    & B\setminus B'
        \arrow[hook]{r}{j_B}
    & B
    \end{mytikzcd}
    \quad\quad
    \begin{mytikzcd}
    D \arrow{d}[left]{b}
    & D_B
        \arrow[hook]{rr}{d_B}
        \ar[d,->>,swap,"b'"]
        \arrow[hook']{l}[above]{d_B}
    && D
        \arrow{d}{a}
    \\
    B & \arrow[hook']{l}[above]{i_B}
    B'
        \arrow{r}{c_A}
    & A\setminus A'
        \arrow[hook]{r}{j_A}
    & A
    \end{mytikzcd}
    \]
    \begin{enumerate}
    \item Put $c_B(a(x)) = b(x)$ for $x\in D_A$. Firstly, $c_B$ is
      well-defined, because \( a(x_1) = a(x_2) \)
      implies $x_2 \in [x_1]_a \subseteq [x_1]_b$ and so $bx_1=bx_2$.
      Secondly, to see that $c_B(\alpha)$ is in $B\setminus B'$,
      assume $b(x) \in B'$ for $x\in D_A$. Then there is $y\in D_B$
      such that $b(y) = b(x)$. Hence $a(x) = a(y)$ because $y \in D_B$
      and $x\in [y]_b$. This leads to the following contradiction:
    \[
        [x]_b = [y]_b \subsetneq
        [y]_a = [x]_a \subseteq [x]_b.
    \]

  \item Analogously, put $c_A(b(x)) = a(x)$ for $x\in
    D_B$.
    Well-definedness is proved similarly as for $c_B$; the image
    restricts to $A\setminus A'$ by the same argument as before but
    with $\subseteq$ and $\subsetneq$ swapped in the last line.
    \end{enumerate}

    Next we consider the diagram below, which commutes by construction
    of $c_A,c_B$:
    \[
    \begin{mytikzcd}[column sep = 17mm]
    HD
        \arrow{rr}{H\fpair{a,b}}
        \arrow{d}[swap]{H\phi}
    && H(A×B)
        \arrow{dd}[sloped,above]{\fpair{H\pi_1,H\pi_2}}
    \\
    H(D_A+D_B)
        \arrow{d}[swap]{H(a'+b')}
    & \begin{array}{c}
        H\big((A'+A\setminus A') \\
        ×(B\setminus B' + B')\big)
    \end{array}
        \arrow[draw=myorange,>->]{ur}[sloped,below]{
        \begin{array}{l}
            H\big([i_A,j_A] \\[-1mm]
            \phantom{H}×[j_B,i_B]\big)
        \end{array}
        }
    \\
    H(A'+B')
        \arrow[draw=myblue,>->]{r}[above,pos=0.7]{\begin{array}{c}
            H\fpair{(A'+c_A),\\[-1mm]
            (c_B + B')}
            \end{array}}
        \arrow[>->,
               draw=myorange,
               end anchor={[yshift=6mm]},
              ]{ur}[sloped,above]{\begin{array}{c}
            \fpair{H(A'+c_A),\\[-1mm]
            H(c_B + B')}
            \end{array}}
    & \begin{array}{c}H(A'+A\setminus A') \\
    × H(B\setminus B' + B')
    \end{array}
    \arrow[draw=myblue,>->]{r}[above]{\begin{array}{c}
        H[i_A,j_A] \\[-1mm] ×H[j_B,i_B]
    \end{array}}
    & HA × HB
    \end{mytikzcd}
    \]
    The left hand morphism at the bottom is monic because $H$ is
    zippable and by Lemma~\ref{zippableExtended}. The second morphism
    is an isomorphism. Thus, the diagonal from $H(A'+B')$ to $H(A×B)$
    is also monic and we can conclude using Lemma~\ref{kernelMono} for
    each of the colored monomorphisms:
    \[
      \ker \big( \fpair{Ha, Hb}\big) = 
        \ker \big(\fpair{H\pi_1,H\pi_2}
            \cdot H\fpair{a,b}\big)\,
        {\color{myblue}=}
        \ker \big(H(a'+b')\cdot H\phi\big)\,
        {\color{myorange}=}
        \ker H\fpair{a,b}.
    \qedhere
    \]
\end{proof}

\paragraph*{Proof of \autoref{optimizationSummary}}

\begin{proposition} \label{kernelComposition}
    Whenever $\ker (a: D\to A) = \ker (b: D\to B)$ then
    $\ker (a\cdot g) = \ker (b\cdot g)$, for $g: W\to D$.
\end{proposition}

\begin{proof}
  The kernel $\ker(a \cdot g)$ can be obtained uniquely
  from $\ker a$ by pasting pullback squares as shown below:
  \[
    \begin{mytikzcd}
    \ker\big( a \cdot g\big)
        \arrow{r}
        \arrow{d}
        \pullbackangle{-45}
    & \bullet
        \pullbackangle{-45}
    \arrow{r}
    \arrow{d}
    & W
    \arrow{d}{g}
    \\
    \bullet
        \arrow{d}
        \arrow{r}
        \pullbackangle{-45}
    & \ker a
        \arrow{d}
        \arrow{r}
        \pullbackangle{-45}
    & D
    \arrow{d}{a}
    \\
    W
    \arrow{r}{g}
    & D
    \arrow{r}{a}
    & a
    \end{mytikzcd}
    \]
    So if $\ker a = \ker b$, then $\ker (a\cdot g) = \ker
    (b\cdot g)$.
\end{proof}
\autoref{optimizationSummary} is immediate from 
\begin{lemma}\label{optimizationCorrectness}
  If $H: \Set\to\Set$ is zippable and $\op{select}$ respects compound blocks,
  then $P_{i+1} = \ker(H\fpair{\bar q_i,q_{i+1}}\cdot \xi) = P_i \cap \ker
  (Hq_{i+1}\cdot \xi)$.
\end{lemma}

\begin{proof}[Proof (\autoref{optimizationCorrectness})] For
  $H$ zippable and $f_i$, $k_i$ as in \autoref{catPT} we have proved:
\begin{itemize}[topsep=1mm,leftmargin=6mm,itemsep=0mm]
\item[]\op{select} respects compound blocks
${\Leftrightarrow}$
    $\ker f_i\cup \ker k_{i+1}$ is a kernel

\item[$\overset{\ref{propZippable}}{\Rightarrow}$]
$ \ker\fpair{Hf_i, Hk_{i+1}} = \ker H\fpair{f_i, k_{i+1}} $
\item[$\overset{\text{\ref{kernelComposition}}}{\Rightarrow}$]
$ \ker(\fpair{Hf_i, Hk_{i+1}}\cdot H\kappa_{P_i}\cdot\xi) = \ker (H\fpair{f_i,
k_{i+1}}\cdot H\kappa_{P_i}\cdot\mathrlap{\xi )}$

\item[$\overset{\phantom{\ref{propZippable}}}{\Rightarrow}$]
    $\ker \big(\fpair{H\bar q_i\cdot \xi, H q_{i+1}\cdot \xi}\big)
    = \ker \big(H\fpair{\bar q_i, q_{i+1}} \cdot \xi\big)
    = \ker \big(H\bar q_{i+1} \cdot \xi\big)
    = P_{i+1}$
\qedhere
\end{itemize}
\end{proof}

\paragraph*{Details for \autoref{sortedCoalgebra}}
Note that any coalgebra $\xi: X \to FGX$ can be trivially decomposed
into $(\xi, \id): (X, GX) \to (FGX, GX)$. (More generally, trivially
decomposing a coalgebra via $\xi \mapsto (\xi, \id)$ and composing a
multi-sorted coalgebra via $(x,y) \mapsto Fy \cdot x$, respectively,
are the object mappings of an adjoint pair of functors, see
\cite{SchroderPattinson11}.)  Furthermore, if
$q: (X,\xi) \epito (Y,\zeta)$ is a quotient coalgebra, one sees that
$(\zeta, \id)$ is a quotient coalgebra of $(x,y)$ via
$(q, Gq \cdot y)$. Conversely, if $(x',y')$ is any quotient coalgebra
of $(x,y)$ via some $(q_1, q_2)$, say, than $Fy' \cdot x'$ is a
quotient coalgebra of $\xi = Fy \cdot x$. Consequently, $\xi$ is a
simple coalgebra iff $(x,y)$ is, and therefore the optimized algorithm
in the multi-sorted setting computes the correct partition for a
composition of \Set-functors.

Now suppose that, $(x',y'): (X',Y') \to (FY',GX')$ is the simple
quotient of $(x,y)$ via $(q_1, q_2): (X,Y) \to (X',Y')$, say. Then,
$(X', Fy'\cdot x')$ is clearly a quotient of $(X, \xi)$ via $q_1$. In order to show simplicity suppose that we have an $FG$-coalgebra morphism $h$ from $(X', Fy'\cdot x')$ to some $(Z, \zeta)$. Then $(h, Gh \cdot y')$ is clearly an $H$-coalgebra morphism from $(x',y')$ to $(\zeta, \id)$ (we consider the two sorts separately below -- the right-hand component is trivial and the left-hand component states that $q$ is an $FG$-coalgebra morphism from $(X', Fy' \cdot x')$ to $(Z, \zeta)$:
\[
  \begin{mytikzcd}
    X' \arrow{r}{x'}
    \arrow{d}{h}
    & 
    FY' \arrow{d}{F(G h \cdot y')}
    \\
    Z \arrow{r}{\zeta}
    &
    FGZ
  \end{mytikzcd}
  \qquad
  \begin{mytikzcd}
    Y' 
    \arrow{r}{y'}
    \arrow{d}{Gh \cdot y'}
    &
    GX' \arrow{d}{Gh}
    \\
    GZ \arrow{r}{\id}
    &
    GZ
  \end{mytikzcd}
\]
Since $(x',y')$ is simple in $\Set^2$ we know that $(h, Gh\cdot y')$ is monic, whence $h$ is injective and we are done.
 
\subsection*{Details for \autoref{sec:efficient}}

\paragraph*{Proof of \autoref{encodingCanonical}}
  Define $E$ as follows. Compose $\xi$ with $\flat$ and the inclusion
  into the set of all maps $A \times Y \to \N$:
  \[
    \begin{mytikzcd}
      X \arrow{r}{\xi}
      & HY \arrow{r}{\flat}
      & \Bagf(A×Y) \descto{r}{=}
      &[-7mm] \N^{(A×Y)} \arrow[hook]{r}
      & \N^{A×Y}.
    \end{mytikzcd}
  \]
  Its uncurrying is a map $\op{cnt}: X×A×Y\to \N$, and we let
  \[
    E :=
    \coprod_{\mathclap{e \in X×A×Y}}
    \op{cnt}(e),
  \]
  where each $\op{cnt}(e) \in \N$ is considered as a finite ordinal
  number. By copairing we then obtain a unique morphism $\op{graph}: E \to X
  \times A \times Y$ defined on the coproduct components as 
  \[
    \op{cnt}(e) \xrightarrow{!} 1 \xrightarrow{e} X \times A \times Y,
  \]
  and we put $\op{type} = H! \cdot \xi$. 
  Note that if $X$ is finite, then so is $E$, since all
  $\flat\cdot\xi(x)$ are finitely supported.

\paragraph*{Details for \autoref{exampleH2}}
\begin{proof}
  We verify \eqref{eqSplitterLabels} for each of the three examples. 

 In general, for $S\subseteq C\subseteq Y$, we have $\op{val}\cdot H\chi_S^C =
\fpair{H\chi_{S}, H\chi_S^C, H\chi_{C\setminus S}}$. Hence to verify the
axiom for $\op{update} = \op{val}\cdot \op{up}$ it suffices to verify that
$\op{up}\cdot\fpair{\op{fil}_S\cdot \flat, H\chi_C } = H\chi_S^C$;
in fact, using $w(C) = H \chi_C$ we have: 
\begin{align*}
  \op{update}\cdot \fpair{\op{fil}_S \cdot\flat, w(C)} &=                           \op{val} \cdot \op{up} \cdot \fpair{\op{fil}_S \cdot\flat, H\chi_C}
  \\
  &= \op{val} \cdot H \chi^C_S \\
  &= \fpair{H\chi_S, H\chi^C_S, H\chi_{C\setminus S}} \\
  &= \fpair{w(S), H\chi^C_S, w(C\setminus S)}.
\end{align*}
\begin{enumerate}
\item For any $f\in HY = G^{(Y)}$, we have:
\begin{align*}
    \op{init}(H!(f), \Bagf\pi_1\cdot\flat(f)) &=
    \big(0, {\textstyle\sum}\Bagf\pi_1\cdot\flat(f)\big)
    \\ & =
    \big(0, \sum_{\mathclap{\substack{y\in Y,\\ f(y)\neq 0}}}
    f(y)\big) =
    \big(0, \sum_{\mathclap{\substack{y\in Y}}} f(y)\big) =
  G^{(\chi_Y)}(f) = w(Y)(f), 
\\
    \op{up}(\op{fil}_S(\flat(f)), H\chi_C(f))
    &=
    \op{up}\big(
        \{f(y)\mid y\in S\},
        \big(\sum_{\mathclap{y\in Y\setminus C}} f(y),
         \sum_{\mathclap{y\in C}} f(y)\big)
         \big)
    \\ &
    = \big(\sum_{\mathclap{y\in Y\setminus C}} f(y),
           \sum_{\mathclap{y\in C}} f(y)- \sum_{\mathclap{y\in S}} f(y),
           \sum_{\mathclap{y\in S}} f(y)
      \big)
    \\ &
    = \big(\sum_{\mathclap{y\in Y\setminus C}} f(y),
           \sum_{\mathclap{y\in C\setminus S}} f(y),
           \sum_{\mathclap{y\in S}} f(y)
      \big)
    = H\chi_S^C(f).
\end{align*}

\item The axiom for $\op{init}$ clearly holds since for any
  $f \in \Dist Y$, we have
  $\groupsum \Bag\pi_1 \cdot \flat(f) = \sum_{y \in Y} f(y) =
  1$. 

  For the axiom for $\op{up}$ the proof is identical as in the
  previous point; in fact, note that for an $f \in \Dist Y$ all
  components of the triple 
  $ 
  \big(\sum_{y\in Y\setminus C} f(y),
  \sum_{y\in C\setminus S} f(y),
  \sum_{y\in S} f(y)
  \big)
  $
  are in $[0,1]$ and their sum is $\sum_{y \in Y} f(y) = 1$. Thus,
  this triple lies in $\Dist 3$ and is equal to $\Dist \chi^C_S(f)$. 

\item For the refinement interface for $\N^{(-)}$ we argue similarly:
  if $f$ lies in $\N^{(Y)}$ then $\sum_{y \in Y} f(y)$ lies in $\N$
  and so do the components of the triple in the proof of the axiom of
  $\op{up}$, whence we obtain $\N^{(\chi^C_S)}(f)$. 

\item Let $t=\sigma(y_1,\ldots,y_n) \in H_\Sigma Y$ with $\sigma$ of arity $n$,
let $c_i = \chi_C(y_i)$, and $I= \{ 1 \le i \le n \mid y_i \in S\}$.
\begin{align*}
    \op{init}(H_\Sigma!(t), \Bagf\pi_1\cdot \flat(t))
    &= \op{init}(\sigma(0,\ldots,0), \Bagf\pi_1(\{(1,y_1),\ldots,(n,y_n) \}))
    \\ &
     = \op{init}(\sigma(0,\ldots,0), \{1,\ldots,n \})
     = \sigma(1,\ldots,1)
    \\ &
     = \sigma(\chi_Y(y_1),\ldots,\chi_Y(y_n))
     = H_\Sigma \chi_Y (t).
     \\
\op{up}(\op{fil}_S\cdot \flat(t), H_\Sigma\chi_C(t)) &=
    \op{up}(\op{fil}_S(\{(1,y_1),\ldots,(n,y_n)\}),
    \sigma(c_1,\ldots,c_n))
    \\ &
    = \op{up}(I, \sigma(c_1,\ldots,c_n))
    \\ &
    = \sigma(c_1 + (1\in I),\ldots,c_i + (i\in I),\ldots, c_n + (n\in I))
    \\ &
    = \sigma(\chi_S^C(y_1),\ldots,\chi_S^C(y_i),\ldots, \chi_S^C(y_n))
    \\ &
    =H_\Sigma\chi_S^C(t).
\end{align*}
In the penultimate step it is used that:
\[
\begin{array}[b]{lcl}
y_i \in Y\setminus C
&\Rightarrow& c_i + (i\in I) =
0 + 0 = 0 = \chi_S^C(y_i)
\\
y_i \in C\setminus S
&\Rightarrow& c_i + (i\in I) =
1 + 0 = 1 = \chi_S^C(y_i)
\\
y_i \in S
&\Rightarrow& c_i + (i\in I) =
1 + 1 = 2 = \chi_S^C(y_i)
\end{array}
\qedhere
\]
\end{enumerate}
\end{proof}

\paragraph*{Information in $w(C)$}
The functor-specific $w(C)$ is not always $H\chi_C$, but always has
at least this information:
\begin{proposition}\label{prop:atleast}
    For any refinement interface, $H\chi_C = H (=1) \cdot \pi_2\cdot
    \op{update}(\emptyset)\cdot w(C)$.
\end{proposition}

\begin{proof}
    The axiom for $\op{update}$ and definition of $\op{fil}_\emptyset$ makes
    the following diagram commute:
    \[
    \begin{mytikzcd}[column sep=1.4cm,baseline=(bot.base)]
    HY
    \arrow[
        rounded corners,
        to path={
            -- ([yshift=-5mm]\tikztostart -| \tikztotarget) \tikztonodes
            -- (\tikztotarget)
        },
    ]{dr}[pos=0.4,sloped,below]{\fpair{w(\emptyset),H\chi_\emptyset^C,
    w(C\setminus \emptyset)}}
    \arrow[
        rounded corners,
        to path={
            (\tikztostart.east)
            -- ([yshift=-4.5mm]\tikztostart -| \tikztotarget) \tikztonodes
            -- (\tikztotarget)
        },
    ]{drr}[pos=0.8,below]{H\chi_\emptyset^C}
    \arrow[
        rounded corners,
        to path={
            (\tikztostart.east)
            -- (\tikztostart -| \tikztotarget) \tikztonodes
            -- (\tikztotarget)
        },
    ]{drrr}[pos=0.8,below]{H\chi_C}
    \arrow[shiftarr={xshift=-17mm}]{d}[left]{\fpair{\emptyset!, w(C)}}
    \arrow{d}[left]{\fpair{\flat\cdot \op{fil}_\emptyset,w(C)}}
    \\
    |[alias=bot]|
    \Bagf A × W
    \arrow{r}{\op{update}}
    & W × H3 × W
    \arrow{r}{\pi_2}
    & H3
    \arrow{r}{H(=1)}
    & H2
    \end{mytikzcd}
    \qedhere
    \]
\end{proof}

\paragraph*{Details for \autoref{examplePowerset}}

\begin{proof} We prove \eqref{eqSplitterLabels} for the refinement
  interface of the finite powerset functor. 
  
The axiom for $\op{init}$ is proved analogously as for $\N^{(-)}$ in
the proof for \autoref{exampleH2}.\ref{exampleH2.1}. 

Note that we have $\op{update} = \fpair{\N^{(=2)}\!,\
{(\geZero 0)}^3\!,\ \N^{(=1)}}\cdot \op{up}\cdot \inj_\B^Y$, where
$\op{up}$ is as for $\N^{(-)}$. Now, we need to show the
commutativity of the diagram below:
\[
\begin{mytikzcd}[column sep=9mm]
    \Potf Y
    \cong \B^{(Y)}
    \arrow[
        to path={
            (\tikztostart.north)
            -| ([yshift=3mm]\tikztostart.north)
            -- ([yshift=3mm]\tikztostart.north  -| emptycell.north east)
            -- (\tikztotarget) \tikztonodes
        },
        rounded corners,
    ]{ddrrr}[sloped,above]{\fpair{\N^{(\chi_S)}\cdot
    \inj_\B^Y,\B^{(\chi_S^C)}, \N^{(\chi_{C\setminus S})}\cdot \inj_\B^Y}}
    \arrow{d}[left]{\fpair{\flat,\N^{(\chi_C)}\cdot \inj_\B^Y}}
    \arrow[hook,bend left]{dr}{\inj_{\B}^Y}
    &
    |[alias=emptycell]| {}
    \\
    \Bagf(\N×Y) × \N^{(2)}
    \arrow{d}[swap]{\op{fil}_S × \N^{(2)}}
    &
    \N^{(Y)}
    \arrow{d}[swap]{\N^{(\chi_S^C)}}
    \arrow{l}[sloped,above,pos=0.5,yshift=1mm]{\fpair{\flat,
    \N^{(\chi_C)}}}
    \arrow[bend left]{rd}[sloped,above]{\fpair{\N^{(\chi_S)},\N^{(\chi^C_S)},\N^{(\chi_{C\setminus S})}}}
    \\
    \Bagf \N × \N^{(2)}
    \arrow{r}[sloped,above]{\op{up}}
    &
    \N^{(3)}
    \arrow{r}[near end]{\begin{array}{c}
        \fpair{\N^{(=2)}, \id, \N^{(=1)}}
        \\[-1mm]
        = \op{val}
        \end{array}}
    & \N^{(2)} \!×\N^{(3)}\!× \N^{(2)}
    \arrow{r}[yshift=2pt]{\id×{(> 0)}^3×\id}
    &[6mm] \N^{(2)} × \B^{(3)} × \N^{(2)}\!.
\end{mytikzcd}
\]
The inner left-hand triangle clearly commutes. The square below it
involving $\op{up}$ and the middle lower triangle commute as shown in
\autoref{exampleH2}.\ref{exampleH2.1}. The first and the third
component of the remaining right-hand part clearly commute, and for
the second component let $f\in \B^{(Y)}$ and compute as follows:
\begin{align*}
\B^{\chi_S^C}(f)
&= \left(
    \bigvee_{y\in Y\setminus C} f(y),
    \bigvee_{y\in C\setminus S} f(y),
    \bigvee_{y\in S} f(y)
\right)
\\
&= \left(
    0 \leZero \sum_{\mathclap{y\in Y\setminus C}} \inj_\B^Y f(y),\quad
    0 \leZero \sum_{\mathclap{y\in C\setminus S}} \inj_\B^Y f(y),\quad
    0 \leZero \sum_{\mathclap{y\in S}} \inj_\B^Y f(y)
\right)
\\
&= (\geZero  0)^3 \cdot \big(g \mapsto \big(
    \sum_{y\in Y\setminus C} g(y),
    \sum_{y\in C\setminus S} g(y),
    \sum_{y\in S} g(y)
\big)\big) \cdot \inj_\B^Y(f)
\\
&= (\geZero  0)^3\cdot \N^{(\chi_S^C)}\cdot \inj_\B^Y(f).
\qedhere
\end{align*}
\end{proof}

\paragraph*{Details for \autoref{exampleTime}}
For all those examples using the $\op{val}$-function from \autoref{exampleH2},
first note that $\op{val}$ runs in linear time (with a constant factor of 3,
because $\op{val}$ basically returns three copies of its input).

For all the monoid-valued functors $G^{(-)}$ for an abelian group, for
$\N$ and for $\Dist$, all the operations, including the summation
$\groupsum e$, run linearly in the size of the input. If the elements $g\in G$
have a finite representation, then so have the elements of $G^{(2)}$ and thus
comparing elements of $g_1,g_2\in G^{(2)}$ is running in constant time.

For a polynomial functor $H_\Sigma$ with bounded arities, we assume that the
name of the operation symbol $\sigma\in \Sigma$ is encoded by a constant-size
integer. So we can also assume that comparison of these integers run in constant
time. Since the signature has bounded arities, the maximum arity available in
$\Sigma$ is independent from the concrete morphism $X\to H_\Sigma Y$, so the
comparison of the arguments of two flat $\Sigma$-terms $t_1,t_2\in H_\Sigma 3$
also runs in constant time.
\begin{enumerate}
\item The first parameter of type $H_\Sigma 1$ can be encoded as
simply a operation symbol $\sigma$. Let $t\in H_\Sigma 2$ be fixed.
Then one explicitly implements
\vspace{2mm}
\[
\op{init}(\sigma, f) =\begin{cases}
    \sigma(\smash{\overbrace{1,\ldots,1}^{\mathclap{\text{arity
    $\sigma$ many}}}})
    & \text{if arity}(\sigma) = |f|
    \\
    t & \text{otherwise.}
\end{cases}
\]
Both the check and the construction of $\sigma(1,\ldots,1)$ are
bounded linearly by the size of $f$. The second case runs in constant
time, since we fixed $t$ beforehand.

\item In $\op{up}(I, \sigma(b_1,\ldots,b_n))$, we can not naively
check all the $1 \in I,\ldots,n\in I$ queries, since this would lead
to a quadratic run-time. Instead we precompute all the queries'
results together.

\begin{algorithmic}[1]
    \State define the array $\op{elem}$ with indices $1\ldots n$, and
    each cell storing a value in $2$.
    \State init \op{elem} to $0$ everywhere.
    \State \textbf{for} $i\in I$ with $i\le n$ \textbf{do} $\op{elem}[i] \gets 1$.
    \State return $\sigma(b_1 + \op{elem}[1], \ldots, b_n +
    \op{elem}[n])$.
\end{algorithmic}
The running time of every line is bound by $|I|+n$.
\end{enumerate}

\paragraph*{Proof of \autoref{lemmaInit}}

\begin{proof}
  The grouping line~7 takes $\CO(|X|\cdot \log|X|)$ time. The first
  loop takes $\CO(|E|)$ steps, and the second one takes $\CO(|X| + |E|)$ time in total over all $x\in
  X$ since $\op{init}$ is assumed to run in linear time. For the invariants:
    \begin{enumerate}
    \item By line \ref{algoLineInitEmptyset}.

    \item After the procedure, for $e_i = x_i\xrightarrow{a_i}y_i, i\in \{1,2\}$,
    $\op{lastW}(e_1) = \op{lastW}(e_2)$
    iff $x_1=x_2$.

    \item This is just the axiom for \op{init} in~\eqref{eqSplitterLabels}, because
    $\unnicefrac{Y}{Q} = \{Y\}$ and $\op{deref}\cdot\op{lastW}(e) =
    \op{init}(\op{type}(x), \Bagf\pi_1\cdot \flat\cdot\xi(x))$.

    \item Since $\ker(H\chi_Y\cdot\xi) = \ker(H!\cdot\xi)$, this is
    just the way $\unnicefrac{X}{P}$ is constructed.\qedhere
    \qedhere
    \end{enumerate}
\end{proof}

\paragraph*{Proof of \autoref{thmSplitCorrect}}

\begin{lemma} \label{p1Properties}
Assume that the invariants hold. Then after part (a) of \autoref{figSplitAlgo}, for the given
$S\subseteq C \in \unnicefrac{Y}{Q}$ we have:
\begin{enumerate}
    \item \label{p1PropToSub}
        For all $x\in X$: $\op{toSub}(x) = \{ e\in E\mid e =
        x\xrightarrow{a}y, y\in S\}$
    \item \label{p1PropFilS}
        For all $x\in X$: $\op{fil}_S\cdot \flat\cdot \xi(x) =
        \Bagf(\pi_2\cdot \op{graph})(\op{toSub}(x))$.

    \item \label{p1PropBS}
        $\op{M}: \unnicefrac{X}{P} \partialto H3$ is a partial map with $M(B) = H\chi_\emptyset^C\cdot \xi(x)$, if there exists an $e = x\xrightarrow{a}y, x\in B, y\in S$, and $M(B)$ is undefined otherwise.

    \item \label{p1PropMark} For each $B\in\unnicefrac{X}{P}$, we have a partial map
    $\op{mark}_B: B\partialto \N$ defined by $\op{mark}_B(x) = \op{lastW}(x)$
    if there exists some $e = x\xrightarrow{a}y$, $y\in S$ and $\op{mark}_B(x)$ is undefined otherwise.

    \item \label{p1PropDeref}
        If defined on $x$, $\op{deref}\cdot\op{mark}_B(x) = w(C,\xi(x))$.

    \item \label{p1PropUnmarked}
        If $\op{mark}_B$ is undefined on $x$, then
        $\op{fil}_S(\flat\cdot \xi(x)) = \emptyset$ and then
        $H\chi^C_S \cdot \xi(x) = H\chi^C_\emptyset \cdot \xi(x)$.

\end{enumerate}
\end{lemma}
\begin{proof}
\begin{enumerate}
\item By lines~2 and~11,
    \begin{align*}
        \op{toSub}(x) &= \{
        e\in \op{pred}(y) \mid y\in S, e = x\xrightarrow{a}y \}
        \\ &= \{
        e\in E \mid y\in S, e = x\xrightarrow{a}y
        \}.
    \end{align*}

\item \(
\begin{aligned}[t]
    \op{fil}_S(\flat\cdot \xi(x))(a)
    &= \sum_{y\in S} (\flat\cdot\xi(x))(a,y)
    = \sum_{y\in S} |\{e\in E\mid e = x\xrightarrow{a} y\}|
    \\ &
    = |\{e\in E\mid e = x\xrightarrow{a} y, y\in S\}|.
\end{aligned}\)

\item By construction $\op{M}$ is defined precisely for those blocks
  $B$ which have at least one element $x$ with an edge
  $e = x \xrightarrow{a} y$ to $S$. Let
  $C = [y]_{\kappa_Q} \in \unnicefrac{Y}{Q}$. Then by
  invariant~\ref{invariant_deref} we know that $M(B)$ is
  \begin{align*}
    \pi_2\cdot \op{update}(\emptyset, \op{deref}\cdot \op{lastW}(e))
    &= \pi_2\cdot \op{update}(\emptyset, w(C,\xi(x)))
    \\ &
    = \pi_2\cdot \op{update}(\op{fil}_\emptyset(\flat\cdot\xi(x)), w(C,\xi(x)))
    \overset{\mathclap{\eqref{eqSplitterLabels}}}{=}
        H\chi_\emptyset^C\cdot\xi(x)
\end{align*}
for some $e=x\xrightarrow{a} y$, $x\in B$, $y\in S$. Since
$\ker(H\chi_\emptyset^C\cdot\xi)= \ker(H\chi_C\cdot\xi)$, invariant
\ref{invariant_Hchi} proves the well-definedness.

\item This is precisely, how $\op{mark}_B$ has been constructed. The
well-definedness follows from invariant \ref{invariant_lastW}. Note
that for any $B$ on which $\op{M}$ is undefined, the list $\op{mark}_B$
is empty.

\item If $\op{mark}_B(x) = p_C$ is defined, then $p_C = \op{lastW}(e)$
for some $e\in \op{toSub}(x)$, and so $\op{deref}(p_C) =
\op{deref}\cdot \op{lastW}(e) = w(C,\xi(x))$ by invariant
\ref{invariant_deref}.

\item If $x\in B$ is not marked in $B$, then $x$ was never contained in line 3. Hence, $\op{toSub}(x) =
\emptyset$, and we have
\[
    \op{fil}_S(\flat\cdot \xi(x))(a)
    = |\{e\in E\mid e = x\xrightarrow{a} y, y\in S\}|
    = |\op{toSub}(x)|
    = 0.
\]
Furthermore we have
\begin{align*}
  H\chi_S^C\cdot \xi(x)
  =&\ %
     \pi_2\cdot \op{update}(\op{fil}_S\cdot \flat\cdot\xi(x), w(C,\xi(x)))
  &\text{by \eqref{eqSplitterLabels}}
  \\=&\ %
       \pi_2\cdot \op{update}(\emptyset,  w(C,\xi(x)))
  &\text{as just shown}
  \\=&\ %
       \pi_2\cdot \op{update}(\op{fil}_\emptyset\cdot \flat\cdot\xi(x), w(C,\xi(x)))
  &\text{by definition}
  \\=&\ %
       H\chi_\emptyset^C\cdot \xi(x)
  &\text{by \eqref{eqSplitterLabels}}
  &\qedhere
\end{align*}
\end{enumerate}
\end{proof}

\begin{lemma} \label{BneqPartialMap} Given $M(B) = v_\emptyset$ in
  line 12. Then after line $\ref{algoLastLine}$, $B_{\neq\emptyset}$
  is a partial map $B \partialto H3$ defined by 
  $B_{\neq\emptyset}(x) = H\chi_S^C\cdot \xi(x)$
  if $H\chi_S^C\cdot \xi(x) \neq v_\emptyset$ and undefined otherwise. 
\end{lemma}
\begin{proof}
  Suppose first that $x\in B$ is marked, i.e.~we have $p_C = \op{mark}_B(x)$ and lines 14--23 are executed. Then we have 
  \begin{align*}
        (w^x_S, v^x, w^x_{C\setminus S}) =\ &  \op{update}(\Bagf(\pi_2\cdot \op{graph})(\op{toSub}[x]),
        \op{deref}[p_C])
        && \text{by line \ref{algoLineUpdate}}
        \\ =\ &
        \op{update}(\op{fil}_S\cdot \flat\cdot \xi(x), w(C,\xi(x)))
        && \text{by \autoref{p1Properties}, \ref{p1PropFilS} and
        \ref{p1PropDeref}}
        \\ =\ &
        (w(S,\xi(x)), H\chi_S^C\cdot \xi(x), w(C\setminus S, \xi(x)))
        && \text{by \eqref{eqSplitterLabels}}.
  \end{align*}
  Thus, if $H\chi_S^C\cdot \xi(x) = v_\emptyset$, $B_{\neq\emptyset}$
  remains undefined because of line~\ref{algoLineIfVXVempty}, and
  otherwise gets correctly defined in line~\ref{algoLastLine}.

  Now suppose that $x\in B$ is not marked. Then by
  \autoref{p1Properties} item \ref{p1PropUnmarked} we know that
  $\op{fil}_S(\flat\cdot\xi(x)) = \emptyset$.  Since
  $(B,v_\emptyset) \in \op{M}$, we know by
  \autoref{p1Properties}.\ref{p1PropBS} that
  $v_\emptyset = H\chi_\emptyset^C\cdot\xi(x')$ for some $x' \in
  B$. Invariant~\ref{invariant_Hchi} then implies that
  $H\chi_\emptyset^C\cdot \xi(x) = H\chi_\emptyset^C\cdot\xi(x')$
  since $\ker(H\chi^C_\emptyset \cdot \xi) = \ker(H\chi_C \cdot
  \xi)$. Then, by \autoref{p1Properties}.\ref{p1PropUnmarked}, we have
  \[
    H\chi_S^C\cdot \xi(x)
    = H\chi_\emptyset^C\cdot \xi(x)
    = H\chi_\emptyset^C\cdot \xi(x')
    = v_\emptyset.
    \qedhere
  \]
\end{proof}

\begin{proof}[Proof of \autoref{thmSplitCorrect}]
    After \autoref{BneqPartialMap}, we know that all $B$ in $\op{M}$
    are refined by $H\chi_S^C\cdot \xi$. Now let $B$ be not in
    $\op{M}$. Then $\op{mark}_B$ is undefined everywhere, so for all $x\in
    B$, we have by \autoref{p1Properties} item \ref{p1PropUnmarked} that 
    $H\chi_S^C\cdot \xi(x) = H\chi_\emptyset^C\cdot \xi(x)$,
    which is the same for all $x\in B$ by invariant \ref{invariant_Hchi}.
    Hence any $B$ not in $\op{M}$ is not split by $H\chi_S^C\cdot\xi$.
\end{proof}

\paragraph*{Proof for \autoref{lemmaAfterSplit}}
\begin{proof}
    We denote the former values of $P,Q, \op{deref}, \op{lastW}$ using
    the subscript $\old$.
    \begin{enumerate}
    \item It is easy to see that $\op{toSub}(x)$ becomes non-empty in
      line 11 only for marked $x$, and for those $x$ it is emptied
      again in line 20.

    \item Take $e_1 = x_1\xrightarrow{a_1}y_1$,  $e_2 =
    x_2\xrightarrow{a_2}y_2$.
    \begin{itemize}
    \item[$\Rightarrow$] Assume $\op{lastW}(e_1) = \op{lastW}(e_2)$.
    If $\op{lastW}(e_1) = p_S$ is assigned in line
    \ref{algoLineSetLastW} for some marked $x$, then $x_1=x_2=x$ and
    $y_1,y_2\in S\in \unnicefrac{Y}{Q}$. Otherwise, $\op{lastW}(e_1) =
    \op{lastW}_\old(e_1)$ and so $\op{lastW}_\old(e_1)=
    \op{lastW}_\old(e_2)$ and the desired property follows from the
    invariant for $\op{lastW}_\old$.

    \item[$\Leftarrow$] If $x_1=x_2$ and $y_1,y_2\in D\in
    \unnicefrac{Y}{Q}$, then we perform a case distinction on $D$. If $D = S$,
    then $\op{lastW}(e_1) = \op{lastW}(e_1) = p_S$. If $D = C\setminus
    S$, then $\op{lastW}(e_1) = \op{lastW}_\old(e_1) =
    \op{lastW}_\old(e_2) =  \op{lastW}(e_2)$. Otherwise, $D \in
    \unnicefrac{Y}{Q_\old}\setminus \{C\}$ and again $\op{lastW}(e_i) =
    \op{lastW}_\old(e_i)$, $i\in \{1,2\}$.

    \end{itemize}

    \item Let $e = x\xrightarrow{a} y$, $D :=
    [y]_{\kappa_Q} \in \unnicefrac{Y}{Q}$ and do a case distinction on
    $D$:
    \[
    \begin{array}{lll}
        D = S
        &\Rightarrow&
            \op{deref} \cdot \op{lastW}(e) = \op{deref}(p_S)
            = w_S^x = w(S,\xi(x)),
        \\
        D = C\setminus S
        &\Rightarrow&
            \op{deref} \cdot \op{lastW}(e) = \op{deref}(p_C)
            = w_{C\setminus S}^x = w(C\setminus S,\xi(x)),
        \\
        D \in \unnicefrac{Y}{Q_\old}\setminus \{C\}
        &\Rightarrow&
            \op{deref} \cdot \op{lastW}(e)
            = \op{deref}_\old \cdot \op{lastW}(e)_\old
            = w(D,\xi(x));
    \end{array}
    \]
    note that the first equation in the second case holds due to lines~\ref{algo10} and~\ref{algoLineMarkedX}.
    For the first two cases note that
    $w_S^x = w(S,\xi(x))$ and
    $w_{C\setminus S}^x = w(C\setminus S, \xi(x))$ by line
    \ref{algoLineUpdate}, \autoref{p1Properties}, items~\ref{p1PropFilS}
    and~\ref{p1PropDeref}, and by the axiom for $\op{update}$ in
    \eqref{eqSplitterLabels} (see the computation in the proof of \autoref{BneqPartialMap}).

    \item Take $x_1,x_2\in B'\in \unnicefrac{X}{P}$ and $D\in
    \unnicefrac{Y}{Q}$ and let $B := [x_1]_{P_\old} = [x_2]_{P_\old}
    \in \unnicefrac{X}{P_\old}$.
    \begin{enumerate}
    \item If $\op{M}(B)$ is defined, then
        $H\chi_S^C\cdot \xi(x_1) = H\chi_S^C\cdot \xi(x_2)$ --
        otherwise they would have been put into different blocks in line
        \ref{algoLineSplit}. So $(x_1,x_2) \in \ker(H\chi_D\cdot
        \xi(x_1))$ is obvious for $D = S$ and $D=C\setminus S$. For
        any other $D \in \unnicefrac{Y}{Q_\old}$,
        $(x_1,x_2) \in \ker (H\chi_D\cdot \xi)$ by the invariant of
        the previous partition.

      \item If $\op{M}(B)$ is undefined, then $\op{mark}_B$ is undefined
        everywhere, in particular for $x_1$ and $x_2$. Then, by
        \autoref{p1Properties} item \ref{p1PropUnmarked} we have 
        $H\chi_S^C\cdot \xi(x_i) = H\chi_\emptyset^C\cdot \xi(x_i)$ for $i = 1, 2$. 

        Since $C\in \unnicefrac{Y}{Q_\old}$ and
        $\ker(H\chi^C_\emptyset \cdot\xi) = \ker(H\chi_C \cdot \xi)$,
        we have
        $H\chi_\emptyset^C\cdot\xi(x_1) =
        H\chi_\emptyset^C\cdot\xi(x_2)$ by invariant
        \ref{invariant_Hchi}, and so
        $(x_1,x_2) \in \ker (H\chi_S^C\cdot \xi)$. By case
        distinction on $D$, we conclude:
    \[
    \begin{array}[b]{lll}
        D = S
        &\Rightarrow&
            (x_1,x_2) \in
            \ker (H\chi_S^C\cdot \xi)
            \subseteq
            \ker (\overbrace{H(=2)\cdot H\chi_S^C}^{H\chi_S=H\chi_D} \cdot \xi)
        \\
        D = C\setminus S
        &\Rightarrow&
            (x_1,x_2) \in
            \ker (H\chi_S^C\cdot \xi)
            \subseteq
            \smash{\ker (\underbrace{H(=1)\cdot H\chi_S^C}_{H\chi_{C\setminus
            S}=H\chi_D} \cdot \xi)}
        \\
        D \in \unnicefrac{Y}{Q_\old}\setminus \{C\}
        &\Rightarrow&
            (x_1,x_2) \in \ker(H\chi_D \cdot\xi),
    \end{array}
    \]
    where the last statement holds by invariant~\ref{invariant_Hchi} for $\unnicefrac{Y}{Q_\old}$. \qedhere
    \end{enumerate}
    \end{enumerate}
\end{proof}

\paragraph*{Proof of \autoref{lemmaTimeSplit}}

\begin{proof}
    The for-loop in \autoref{algoPredCollecting} has $\sum_{y\in S}
    |\op{pred}(y)|$ iterations, each consisting of constantly many
    operations taking constant time.
    Since each loop appends one element to some initially empty
    $\op{toSub}(x)$, we have
    \[
        \sum_{y\in S}|\op{pred}(y)| = \sum_{x\in X} |\op{toSub}(x)|
    \]
    In the body of the for-loop in line \ref{algoLineMarkedX}, the
    only statements not running in constant time are
    $\ell \gets \Bagf(\pi_2\cdot \op{graph})(\op{toSub}(x))$ (line
    \ref{algoLineLabels}), $\op{update}(\ell,\op{deref}(p_C))$ (line
    \ref{algoLineUpdate}), and the loop in line
    \ref{algoLineSetLastW}. The time for each of them is linear in the
    length of $\op{toSub}(x)$. The loop in line \ref{algoLineMarkedX}
    has at most one iteration per $x\in X$.  Hence, since each $x$ is
    contained in at most one block $B$ from line~\ref{algoSplitLoop},
    the overall complexity of line \ref{algoSplitLoop} to
    \ref{algoLastLine} is at most
    \( \sum_{x\in X} |\op{toSub}(x)| = \sum_{y\in S}|\op{pred}(y)| \),
    as desired.
\end{proof}

\paragraph*{Details on the overall complexity for sortings in line
\ref{algoLineSplit}}

\begin{remark}
Recall that when grouping $Z$ by $f: Z\to Z'$ one calls an element
$p\in Z'$ a \emph{possible majority candidate} (PMC) if either
\begin{equation}
|\{z\in Z\mid f(z) = p\}|
\ge |\{z\in Z\mid f(z) \neq p\}|
\label{pmcproperty}
\end{equation}
or if no element fulfilling \eqref{pmcproperty} exists. A PMC can be
computed in linear time \cite[Sect.~4.3.3]{Backhouse1986}. When
grouping $Z$ by $f$ using a PMC, one first determines a PMC $p\in
Z'$, and then only sorts and groups $\{z\mid f(z) \neq
p\}$ by $f$ using an $n\cdot \log n$ sorting algorithm.
\end{remark}
\begin{lemma} \label{SplitComplexity} Summing over all iterations, the
total time spent on grouping $B_{\neq\emptyset}$ using a PMC is in
$\CO(|E|\cdot \log|X|)$.
\end{lemma}
The proof is the same as in the weighted setting of Valmari and
Franceschinis~\cite[Lemma~5]{ValmariF10}. For the convenience of the
reader, an adaptation to our setting is provided:
\begin{proof}
Formally we need to prove that for a family $S_i\subseteq C_i\in
\unnicefrac{Y}{Q_i}$, $1\le i \le k$, the overall time spend on
grouping the $B_{\neq\emptyset}$ in all the runs of $\textsc{Split}$
is in total $\CO(|E|\cdot \log|X|)$.

First, we characterize the subset $\Middleblock_B \subseteq B_{\neq\emptyset}$
of elements that have edges into both $S_i$ and $C_i\setminus S_i$. In the
second part, we show that if we assume each sorting step of $B_{\neq\emptyset}$
is bound by $2\cdot |\Middleblock_B|\cdot \log(2\cdot |\Middleblock_B|)$, then
the overall complexity is as desired. In the third part, we use a PMC to argue
that sorting each $B_{\neq\emptyset}$ is indeed bounded as assumed. Since we
assume that comparing two elements of $H3$ runs in constant time, the time
needed for sorting amounts to the number of comparisons needed while sorting,
i.e.~$\CO(n\cdot \log n)$ many.
\begin{enumerate}
\item
    For a $(B,v_\emptyset)\in \op{M}$ in the $i$th iteration consider
    $B_{\neq\emptyset}$. We define
    \begin{itemize}
    \item the \emph{left block} $\Leftblock_B^i := \{
        x \in B \mid
        H\chi_{C_i}^{C_i}\cdot \xi(x) = H\chi_{S_i}^{C_i}\cdot \xi(x) \neq
        H\chi_{\emptyset}^{C_i}\cdot \xi(x)
    \}$, and

    \item the \emph{middle block} $\Middleblock_B^i := \{
        x \in B \mid H\chi_{C_i}^{C_i}\cdot \xi(x) \neq
        H\chi_{S_i}^{C_i}\cdot \xi(x)
        \neq H\chi_{\emptyset}^{C_i}\cdot \xi(x)
    \}$.
    \end{itemize}
    Since $B_{\neq\emptyset}(x)$ is defined iff $H\chi_{S_i}^{C_i}\cdot\xi(x)
    \neq v_\emptyset$, and since $v_\emptyset = H\chi_\emptyset^C\cdot \xi(x)$
    holds by \autoref{p1Properties} item~\ref{p1PropBS}, the domain of
    $B_{\neq\emptyset}$ is $ \Leftblock_B^i \cup\Middleblock_B^i$. Any $x\in B$
    with no edge to ${S_i}$ is not marked, and so $H\chi_{S_i}^{C_i}\cdot\xi(x)
    = H\chi_\emptyset^C\cdot \xi(x)$, by \autoref{p1Properties}
    item~\ref{p1PropUnmarked}; by contraposition, any $x\in
    \Leftblock_B^i\cup\Middleblock_B^i$ has some edge into ${S_i}$. We can make
    a similar observation for $H\chi_{C_i}^{C_i}\cdot\xi$. If $x\in B$ has no
    edge to ${C_i}\setminus {S_i}$, then $\op{fil}_{S_i}(\flat\cdot\xi(x)) =
    \Bagf\pi_1\cdot\flat\cdot\xi(x)$ by the definition of $\op{fil}_S$, and
    therefore we have:
    \[
      \begin{array}{r@{}c@{}l}
        H\chi_{C_i}^{C_i}\cdot\xi(x)
        &\stackrel{\eqref{eqSplitterLabels}}{=}& \pi_2\cdot\op{update}(\op{fil}_{C_i}(\flat\cdot\xi(x)),
        w({C_i}))
        \overset{\text{Def.}}{=}
        \pi_2\cdot\op{update}(\Bagf\pi_1\cdot\flat\cdot\xi(x), w({C_i}))
        \\
        &=& \pi_2\cdot\op{update}(\op{fil}_{S_i}(\flat\cdot\xi(x)),
        w({C_i}))
        \stackrel{\eqref{eqSplitterLabels}}{=}
        H\chi_{S_i}^{C_i}\cdot\xi(x).
      \end{array}
    \]
    By contraposition, all $x\in \Middleblock_B^i$ have an edge to
    ${C_i}\setminus {S_i}$.

    Note that $\ker(H\chi_{C_i}^{C_i}\cdot \xi) =
    \ker(H\chi_{C_i}\cdot \xi)$. Since $\Leftblock_B^i \subseteq B\in
    \unnicefrac{X}{P_i}$ and $C_i\in\unnicefrac{Y}{Q_i}$, invariant
    \ref{invariant_Hchi} provides that we have an $\ell^i_B \in H3$
    such that $\ell^i_B = H\chi_{C_i}^{C_i}\cdot \xi(x)$ for all $x\in
    \Leftblock_B^i$. Hence, we have obtained $\ell^i_B \in H3$ such
    that
    \[
        \Leftblock_B^i = \{x\in B_{\neq\emptyset} \mid H\chi_{S_i}^{C_i}\cdot\xi(x) =
        \ell^i_B \}\quad\text{and}\quad
        \Middleblock_B^i = \{x\in B_{\neq\emptyset} \mid
        H\chi_{S_i}^{C_i}\cdot\xi(x) \neq \ell^i_B \}.
    \]

  \item Let the number of blocks to which $x$ has an edge be denoted by
    \begin{align*}
        \sharp^i_Q(x)
        &=
        |\{
            D\in \unnicefrac{Y}{Q_i}\mid e = x\xrightarrow{a}{y}, y\in
            D
        \}|.
    \end{align*}
    Clearly, this number is bounded by the number of outgoing edges of
    $x$, i.e.~$\sharp^i_Q(x) \le |\flat\cdot\xi(x)|$, and so
    \[
        \sum_{x\in X} \sharp_Q^i(x) \le
        \sum_{x\in X} |\flat\cdot\xi(x)| = |E|.
    \]
    Define
    \begin{align*}
        \sharp_\Middleblock^i(x)
        &= |\{
            0 \le j \le i \mid
            x\text{ is in some }\Middleblock_B^j
        \}|.
    \end{align*}
    If in the $i$th iteration $x$ is in the middle block
    $\Middleblock_B^i$, then $\sharp_Q^{i+1}(x) = \sharp_Q^{i}(x) + 1$,
    since $x$ has both an edge to $S_i$ and $C_i\setminus S_i$.
    It follows that for all $i$, $\sharp_\Middleblock^i(x) \le \sharp_Q^i(x)$, and
    \[
        \sum_{x\in X} \sharp_\Middleblock^i(x)
        \le \sum_{x\in X} \sharp_Q^i(x)
        \le |E|.
    \]
    Let $T$ denote the total number of middle blocks $\Middleblock_B^i$, $1\le i
    \le k$, $B$ in the $i$th $\op{M}$, and let $\Middleblock_t$, $1\le t \le T$,
    be the $t$th middle block. The sum of the sizes of all middle blocks is the
    same as summing up, how often each $x\in X$ was contained in a middle block,
    i.e.
    \[
        \sum_{t=1}^T |\Middleblock_t|
        = \sum_{x\in X} \sharp_\Middleblock^k(x)
        \le |E|.
    \]
    Using the previous bounds and the obvious $|\Middleblock_t| \le |X|$, we now obtain
    \begin{align*}
        &\phantom{\le\ }\sum_{t=1}^T
        2 \cdot |\Middleblock_t|\cdot \log (2\cdot|\Middleblock_t|)
        \le
        \sum_{t=1}^T
        2 \cdot |\Middleblock_t|\cdot \log (2\cdot |X|)
        = 2 \cdot \big(\sum_{t=1}^T |\Middleblock_t|\big) \cdot \log(2\cdot
        |X|)
        \\
        & \le 2\cdot |E|\cdot \log(2\cdot |X|)
        = 2 \cdot |E|\cdot \log(|X|) + 2 \cdot |E|\cdot \log(2)
        \in \CO\big(|E|\cdot \log(|X|)\big).
    \end{align*}

    \item We prove that sorting $B_{\neq\emptyset}$ in the $i$th
    iteration is bound by $2\cdot |\Middleblock_B^i|\cdot \log(2\cdot
    |\Middleblock_B^i|)$ by case distinction on the possible majority
    candidate:
    \begin{itemize}
    \item If $\ell_B^i$ is the possible majority
    candidate, then the sorting of $B_{\neq\emptyset}$ sorts
    precisely $\Middleblock_B^i$ which indeed amounts
    to
    \[
        |\Middleblock_B^i|\cdot \log(|\Middleblock_B^i|)
        \le 2\cdot |\Middleblock_B^i|\cdot \log(2\cdot
        |\Middleblock_B^i|).
    \]
    \item If $\ell_B^i$ is not the possible
    majority candidate, then $|\Leftblock_B^i| \le
    |\Middleblock_B^i|$.
    In this case sorting $B_{\neq\emptyset}$ is bounded by
    \[
        (|\Leftblock_B^i|+|\Middleblock_B^i|)\cdot
        \log(|\Leftblock_B^i|+|\Middleblock_B^i|)
        \le 2\cdot |\Middleblock_B^i|\cdot \log(2\cdot
        |\Middleblock_B^i|).
        \qedhere
    \]

    \end{itemize}
\end{enumerate}
\end{proof}

\paragraph*{Proof of \autoref{lemmaTime}}
\begin{proof}
\begin{enumerate}
\item Clearly $Q_{i+1}$ is finer than $Q_i$. Moreover, since
  $\chi^{C_i}_{S_i}$ merges all elements of $S_i$ we have
  $S_i\in\unnicefrac{Y}{Q_{i+1}}$. For any $i < j$ with $y\in S_i$ and
  $y\in S_j$, we know that $C_j \subseteq S_i$ since $C_j$ is the
  block containing $y$ in the refinement $\unnicefrac{Y}{Q_j}$ of
  $\unnicefrac{Y}{Q_{i+1}}$ in which $S_i$ contains $y$. Hence, we
  have $2\cdot|S_j| \le |C_j|\le |S_i|$. Now let $i_1 < \ldots < i_n$
  be all the elements in $\{i < k \mid y \in S_i\}$. Since
  $y\in S_{i_1},\ldots, y\in S_{i_n}$, we have
  $2^n \cdot |S_{i_n}| \le |S_{i_1}|$.  Thus
  \( |\{i < k \mid y \in S_i\}| = n = \log_2(2^n) \le \log_2(2^n\cdot
  |S_{i_n}|) \le \log_2|S_{i_1}| \le \log_2 |Y| \), where the last
  inequality holds since $S_{i_1} \subseteq Y$.

\item In the $\CO$ calculus we have as the total time complexity:
\begin{align*}
    \sum_{0\le i < k}
    \sum_{y\in S_i} |\op{pred}(y)|
    &=
    \sum_{y\in Y}
    \sum_{\substack{0\le i < k \\  S_i \ni y}}
    |\op{pred}(y)|
    =
    \sum_{y\in Y}
    \big(
    |\op{pred}(y)|
    \cdot
    \sum_{\substack{0\le i < k \\ S_i \ni y}} 1
    \big)
    \\ &
    \le
    \sum_{y\in Y}
    \big(
    |\op{pred}(y)|
    \cdot
    \log|Y|
    \big)
    =
    \big(
    \sum_{y\in Y}
    |\op{pred}(y)|
    \big)
    \cdot
    \log|Y|
    \\ &
    =
    |E|\cdot \log|Y|,
\end{align*}
where the inequality in the second line holds by the first part of our lemma.\qedhere
\end{enumerate}
\end{proof}

\paragraph*{Details for Example~\ref{ex:comp}}
There are two ways to handle the functors of type $H=\mathcal{I}\times G$.
The first way is to modify the functor interface as follows:
\[
  \begin{array}{lll}
    G1 &\mapsto& \mathcal{I}\times G1 \\
    W &\mapsto& \mathcal{I}\times W \\
  \end{array}\qquad\qquad \begin{array}{lll}
    \flat &\mapsto& \mathcal{I}\times GY \xrightarrow{\pi_2} GY
                    \xrightarrow{\flat} \Bagf(A\times Y) \\
    \op{init} &\mapsto& \id_\mathcal{I}\times \op{init} \\
  \end{array}
  \]
  $\op{update}$ is replaced by the following
  \[
    \begin{mytikzcd}[column sep=3mm]
         \Bagf(A) \!\times\! (\mathcal{I}\!\times\! W)
         \arrow{r}{\cong\ \ }
         & (\Bagf(A) \!\times\! W) \!\times\! \mathcal{I} 
         \arrow{r}[yshift=2mm]{\op{update}\times \mathcal{I}}
         & (W\!\times\! G3\!\times\! W)\times \mathcal{I}
         \arrow{r}
         & ((\mathcal{I}\!\times\! W)\times
           (\mathcal{I}\!\times\! G3)\times (\mathcal{I}\!\times\! W))
      \end{mytikzcd}
\]
where the first and the last morphism are obvious.
The second
way is to decompose the functor into $\mathcal{I}\times(-)$ and $G$; this only
introduces one new edge per state holding the state's value in $\mathcal{I}$.
Both ways do not affect the complexity in the $\CO$ notation. This applies to
the examples of unlabeled transition systems and weighted systems.
\begin{enumerate}
\addtocounter{enumi}{1} 
\addtocounter{enumi}{1} 
\item  When decomposing a coalgebra $X\to \Potf(A\times X)$ (with $|E|$ edges)
  into a multisorted coalgebra for $\Potf$ and $A\times(-)$, the new sort $Y$
  contains one element per edge. So the multisorted coalgebra has $|X|+|E|$
  states and still $|E|$ edges, leading to a complexity of $\CO((|X|+|E|)\cdot
  \log (|X|+|E|))$.

\item To obtain the alphabet size as part of the input to our
  algorithm we consider DFAs as labelled transition systems encoding
  the letters of the input alphabet as natural numbers, i.e.~coalgebras for $X \to 2 \times
  \Potf(\N \times X)$. Decomposing the type functor into $F = 2
  \times \Potf$ and $G = \N \times (-)$ we equivalently get a
  coalgebra $(X, Y) \to (FY, GX)$ over $\Set^2$. 

\item If you have a segala system as a coalgebra $\xi: X\to \Potf(A\times \Dist X)$,
  then Baier, Engelen, and Majster-Cederbaum~\cite{BaierEM00} define the number of
  states and edges respectively as
  \[
    n = |X|,\qquad
    m_p = \sum_{x\in X} |\xi(x)|.
  \]
  The decomposition of this coalgebra results in maps
  \[
    p: X\to \Potf Y
    \qquad
    a: Y\to A\times Z
    \qquad
    d: Z\to \Dist X.
  \]
  The system $\xi$ mentions at most $m_p$ distributions, so $|Y|,|Z| \le m_p$.
  The maps $p$ and $a$ need at most $m_p$ edges, and let $m_d$ denote the number
  of edges needed to encode $d$. Then we have $n + 2\cdot m_p$ states and
  $2\cdot m_p + m_d$ edges, and thus we have a complexity of
  \[
    \CO((m_p+m_d+n)\cdot \log(m_p+n))
  \]
  whereas the complexity in \cite{BaierEM00} is
  \[
    \CO((m_p\cdot n)\cdot (\log m_p + \log n))).
  \]
  
\end{enumerate}

\subsection*{Details for \autoref{sec:conc}}

\begin{example}
    The monotone neighbourhood functor mapping a set $X$ to
    \[
        \mathcal{M}(X) =
        \{
            N \subseteq \Pot X
            \mid
            A \in N \wedge B\supseteq A
            \implies B\in N
        \}
    \]
    is not zippable. There are neighbourhoods which are identified by
    $\op{unzip}$; indeed let $A = \{a_1, a_2\}$, $B = \{b_1, b_2\}$
    and denote by $(-)\upa$ the upwards-closure:
    \begin{align*}
        \op{unzip}\left(\left\{
            \{a_1,b_1\},
            \{a_2,b_2\}
        \right\}\upa\right)
        &=
        \big(
            \big\{\{a_1,*\},
            \{a_2,*\}\big\}\upa
        ,
            \big\{\{*,b_1\},
            \{*,b_2\}\big\}\upa
        \big)
        \\&=
        \op{unzip}\left(\left\{
            \{a_1,b_2\},
            \{a_2,b_1\}
        \right\}\upa\right).
    \end{align*}
\end{example}

\end{document}